\documentclass{IEEEtran}
\IEEEoverridecommandlockouts
\usepackage{cite}
\usepackage{amsmath,amssymb,amsfonts}
\usepackage{algorithmic}
\usepackage{graphicx}
\usepackage{textcomp}
\usepackage{color}
\usepackage{verbatim}
\usepackage{amsthm}
\usepackage{mathtools}
\usepackage{subfigure}
\usepackage{enumitem}  
\graphicspath{{Graphics}}
\usepackage[normalem]{ulem}
\usepackage{cases}
\usepackage{steinmetz}
\usepackage{geometry}
\newtheorem{lemma}{Lemma}

\DeclarePairedDelimiter{\ceil}{\lceil}{\rceil}
\DeclarePairedDelimiter\floor{\lfloor}{\rfloor}

\DeclareMathOperator{\rem}{\mathrm{mod}}
\DeclareMathOperator{\sign}{sign}

\newcommand{\rom}[1]{\uppercase\expandafter{\romannumeral #1\relax}}
\newcommand{\Fo}[1]{\color{red}#1 \color{black}}
\newcommand{\Ch}[1]{\color{magenta}#1 \color{black}}

\usepackage{fancyhdr}
\usepackage[yyyymmdd,hhmmss]{datetime}
\usepackage{hyperref}
\hypersetup{
    colorlinks=true,
    linkcolor=blue,
    filecolor=magenta,      
    urlcolor=cyan,
    }

\newcommand{\fig}{Fig.}    
\newcommand{\tab}{Table}  
\newcommand{\secR}{Section}
\title{Antenna Combiner for Periodic Broadcast V2V Communication Under Relaxed Worst-Case Propagation\\
\thanks{This research has been carried out in the antenna systems center \emph{ChaseOn} in a project financed by Swedish Governmental Agency of Innovation Systems (Vinnova), Chalmers, Bluetest, Ericsson, Keysight, RISE, Smarteq, and Volvo Cars.}%
\thanks{The authors are with the Communication Systems Group, Department of Electrical Engineering, Chalmers University of Technology, 412 96 Gothenburg, Sweden (e-mail: chouaib@chalmers.se; erik.strom@chalmers.se; fredrik.brannstrom@chalmers.se)}%
}%

\newcommand{\psiZ}{\hat{\psi}_1}  

\newcommand{\Apl}{A}  
\newcommand{\Dpl}{d}
\newcommand{\Dplz}{d_{\mathrm{ref}}}

\newcommand{\cJ}{c_1} 
\newcommand{\cJJ}{c_2} 

\newcommand{\ph}{\alpha}  
\newcommand{\phI}{\beta}  
\newcommand{\phV}{\boldsymbol{\alpha}}

\newcommand{\om}{\Omega} 
\newcommand{\ps}{\psi}  
\newcommand{\psV}{\boldsymbol{\psi}}  
\newcommand{\phS}{\tilde{\alpha}} 
\newcommand{\phBar}{\bar{\alpha}} 

\newcommand{\Astar}{\mathcal{A}^{\star}} 

\newcommand{\bJ}{b} 
\newcommand{\aJ}{a}

\newcommand{\LPL}{L_{\mathrm{PL}}}

\newcommand{\LA}{L_{\mathrm{\Omega}}}

\newcommand{\aom}{a_{\mathrm{\Omega}}}
\newcommand{\bom}{b_{\mathrm{\Omega}}}

\newcommand{\apl}{a_{\mathrm{PL}}}
\newcommand{\bpl}{b_{\mathrm{PL}}}

\newcommand{\aphi}{a_{\mathrm{\phi}}}
\newcommand{\aphiz}{b_{\mathrm{\phi}}}

\newcommand{\agph}{a_{\mathrm{PH}}}
\newcommand{\bgph}{b_{\mathrm{PH}}}

\newcommand{\gr}{g^{\mathrm{r}}}  
  
\newcommand{\omt}{\Omega^{\mathrm{s}}}  
\newcommand{\omr}{\Omega^{\mathrm{r}}}  
\newcommand{\phr}{\alpha^{\mathrm{r}}}

\newcommand{\phrV}{\boldsymbol{\alpha}^{\mathrm{r}}}

\newcommand{\psrV}{\boldsymbol{\psi}^{\mathrm{r}}}

\newcommand{\phrI}{\beta^{\mathrm{r}}}

\newcommand{\Gr}{G_{\mathrm{r}}}

\newcommand{\psr}{\psi^{\mathrm{r}}}  
\newcommand{\pst}{\psi^{\mathrm{s}}}

\newcommand{\Lt}{L_{\mathrm{s}}}
\newcommand{\Lr}{L_{\mathrm{r}}}

\newcommand{\phir}{\phi^{\mathrm{r}}}

\newcommand{\Tm}{{T_\mathrm{m}}}

\newcommand{\SsnrA}{S_\mathrm{\Omega}}

\newcommand{\JA}{J_\mathrm{\Omega}}

\newcommand{\Xstar}{\mathcal{X}^\star}

\newcommand{\Nc}{N_{\textrm{c}}}

\newcommand{\dx}{d_{\mathrm{x}}}
\newcommand{\Lw}{d_{\mathrm{y}}}

\newcommand{\dv}{\Delta_{\mathrm{v}}}
\newcommand{\da}{\delta_{\mathrm{a}}}

\newcommand{\daz}{\delta_{\mathrm{a},0}}

\newcommand{\Iq}{\mathcal{I}}

\newcommand{\SsnrPL}{S_{\mathrm{PL}}}
\newcommand{\SsnrPH}{S_{\mathrm{\phi}}}
\newcommand{\phim}{\phi_{\mathrm{min}}}

\newcommand{\JPL}{J_{\mathrm{PL}}}
\newcommand{\JPH}{J_{\mathrm{\phi}}}

\usepackage{glossaries}
\makeglossaries
\newacronym{AoI}{AoI}{age-of-information}
\newacronym{IRT}{IRT}{inter-reception time}
\newacronym{BrEP}{BrEP}{burst error probability}
\newacronym{PEP}{PEP}{packet error probability}
\newacronym{SNR}{SNR}{signal-to-noise ratio}
\newacronym{SBR}{SBR}{single bounce reflection}
\newacronym{LOS}{LOS}{line of sight}
\newacronym{AOA}{AOA}{angle of arrival}
\newacronym{AOD}{AOD}{angle of departure}
\newacronym{ACN}{ACN}{analog combining network}
\newacronym{ASN}{ASN}{antenna switching network}
\newacronym{ABN}{ABN}{analog beamforming network}
\newacronym{VU}{VU}{vehicular user}
\newacronym{CAM}{CAM}{cooperative awareness message}
\newacronym{MRC}{MRC}{maximal ratio combining}
\newacronym{RF}{RF}{radio frequency}
\newacronym{CSI}{CSI}{channel state information}
\newacronym{C-ITS}{C-ITS}{cooperative intelligent transportation systems}
\newacronym{MIMO}{MIMO}{multiple-input multiple-output}
\newacronym{V2V}{V2V}{vehicle-to-vehicle}
\newacronym{Tx}{Tx}{transmitter}
\newacronym{Rx}{Rx}{receiver}
\newacronym{OFDM}{OFDM}{orthogonal frequency division multiplexing}
\newacronym{MP}{MP}{multipath} 
\newacronym{PL}{PL}{path-loss} 
\newacronym{CDD}{CDD}{cyclic delay diversity}

\newacronym{LS}{LS}{least squares}

\definecolor{electricpurple}{rgb}{0.75, 0.0, 1.0}
\definecolor{flamingopink}{rgb}{0.99, 0.56, 0.67}
\definecolor{c_orange}{rgb}{1.0, 0.8, 0.361}
\definecolor{c_Sb}{rgb}{1.0, 0.8, 0.361}

\definecolor{flame}{rgb}{0.89, 0.35, 0.13}
\graphicspath{{Graphics}}
\usepackage{glossaries}
\makeglossaries

\begin{document}
\glsdisablehyper
\twocolumn
\author{Chouaib~Bencheikh~Lehocine, Erik~G.~Str{\"{o}}m,~\IEEEmembership{Fellow,~IEEE,} and Fredrik~Br{\"{a}}nnstr{\"{o}}m}%
	\maketitle
\begin{abstract}

The performance of a previously developed \gls{ACN} of phase shifters for vehicle-to-vehicle communication is investigated. The original \gls{ACN} was designed to maximize the sum of the \glspl{SNR} for $K$ consecutive, broadcast, periodic cooperative awareness messages when communication is over a dominant path whose \gls{AOA} is constant over the duration of $K$ packets. In this work, we relax this scenario by allowing the dominant path \gls{AOA} and \gls{PL} to be time-variant. Assuming a highway scenario with \gls{LOS} propagation between vehicles, we use affine approximations to model the time variation of different path quantities, including the \gls{PL}, the relative distance-dependent phase shift between antennas, and the \gls{AOA}-dependent far-field function of the antennas. Using these approximations, we analytically derive the \gls{ACN} sum-\gls{SNR} as each one of these quantities vary over the duration of $K$ packets. Moreover, we suggest a phase slope design rule that is robust against time variation of the dominant path and optimal under time-invariant conditions. Finally, we validate this design rule using numerical computations and an example of vehicular communication antenna elements.



\end{abstract}
\glsresetall
\glsunset{C-ITS}
\section{Introduction}
\IEEEPARstart{C}{ooperative} intelligent transportation system (C-ITS) rely on the exchange of \glspl{CAM} to increase traffic safety and efficiency on roads. 
\glspl{CAM} are all-to-all broadcast, periodic packets carrying status information about the dynamics of the disseminating vehicles.  Due to the broadcast nature of these messages, vehicles need to have an antenna system with good gain in all directions. This is challenged by the fact that antenna patterns are distorted by vehicles body, mounting positions, the housing of antennas, etc. These factors cause antennas to have a low gain or even blind spots in certain directions~\cite{AntPlac2007,AntPlac_2_2014,AntPlac2010}. 
In the case that the transmitted and/or received signal coincides with directions where the antenna system has low gain then the cooperative service can be compromised. Therefore, in~\cite{ACN}, it is proposed to process multiple antennas using an \gls{ACN} of phase shifters to mitigate the vehicle-body distortions, and enable omnidirectional coverage at the receive side. \gls{ACN} is a low-cost, low-complexity solution that does not rely on \gls{CSI}, unlike the classical digital counterparts, selection combining, equal gain combining, and \gls{MRC}.
In~\cite{lehocine2020HC}, the fully analog \gls{ACN} combining was enhanced using an \gls{MRC}-based digital stage, to form a hybrid-combining scheme.
Then, a transmit-side counterpart to \gls{ACN}, namely, an \gls{ABN} of phase shifters has been developed in~\cite{lehocine2021abn}. \gls{ABN} does not rely on \gls{CSI} either, and is fully analog, and hence has a lower cost than digital schemes, like \gls{CDD}~\cite{CDD}, and Alamouti~\cite{Alamouti}.


To assess the performance of the multiple antenna schemes in~\cite{ACN,lehocine2020HC,lehocine2021abn} \gls{BrEP} has been used. \gls{BrEP} is related to the reliability at a \gls{C-ITS} application level, where an outage occurs if the available information about a certain vehicle at a receiving end is outdated after a loss of $K$ consecutive \glspl{CAM}. Using certain assumptions, minimizing \gls{BrEP} is found~\cite{ACN,lehocine2020HC,lehocine2021abn} to be equivalent to maximizing the sum-\gls{SNR} of $K$ consecutive packets, and it is the performance metric used in the design of the schemes.
To ensure robust \gls{V2V} communication, \gls{ACN} and \gls{ABN} have been designed assuming a worst-case propagation environment, corresponding to a single dominant path with an \gls{AOA} and an \gls{AOD} that are negligibly varying over the duration of $K$ consecutive packets. Consequently, the distance-dependent relative phase shifts between the antennas, and \gls{PL} were assumed to be negligibly varying over the same duration.  A set of phase slopes that maximize the sum-\gls{SNR} of $K$ consecutive \glspl{CAM} when the dominant path direction coincides with the worst-case \gls{AOA} and \gls{AOD} of the antenna systems have been derived in~\cite{ACN,lehocine2020HC,lehocine2021abn}. 

The negligibly varying single dominant path assumption holds under certain positions and speeds of the \gls{Tx} and the \gls{Rx}, and potential reflecting objects. Therefore, we are interested in investigating the performance of the designed multiple antenna system when the geometries and the mobility of the \gls{Tx} and \gls{Rx} result in a time-varying (instead of an unvarying) single dominant path. Namely, we are interested in investigating the performance of the system when the \gls{AOA}, \gls{AOD}, the relative phase shifts between the antennas, and the \gls{PL} are time-varying over the duration of $K$ packets.

The optimal phase slopes under worst-case propagation assumption~\cite{ACN,lehocine2021abn} are available for a generic $\Lt\times\Lr$ \gls{ABN}-\gls{ACN} where $\Lt$, and $\Lr$ are the number of transmit and receive antennas, respectively. However, for such a system, it is difficult to model, analytically study, and understand the effects of time variation of the \gls{AOA}, the \gls{AOD}, the relative phase shift between antennas, and the \gls{PL}. Therefore, we study the effects of these quantities for a $1\times 2$ \gls{ACN} system (equivalent to $2\times1$ \gls{ABN}). In particular, we first approximate, based on a reference highway scenario with a \gls{LOS} propagation, the time variation of these quantities at medium and large distances between the \gls{Tx} and \gls{Rx} using an affine function. Besides distance, antenna separation and speeds are taken into account in the time variation modeling. Note that a dominant path is not necessarily a \gls{LOS}, it can be a reflected path too. However, we assume a \gls{LOS} propagation since it simplifies the modeling of the quantities of interest. Second, we analytically derive the loss in sum-\gls{SNR} incurred on the \gls{ACN} system for different phase slopes when the dominant path quantities are time-varying. Third, we derive a design rule to choose a robust phase slope that sustains good performance when the \gls{AOA}, the phase shifts, and the \gls{PL} vary over the duration of $K$ \glspl{CAM}. Finally, using numerical computation, we visualize the performance of the \gls{ACN} and validate the design rule.
The investigation of the $1\times 2$ \gls{ACN} system serves as a guideline for investigating $1\times \Lr$, $\Lr>2$, and $\Lt\times \Lr$ systems performance. The loss functions in sum-\gls{SNR} can be numerically evaluated (they seem analytically intractable), then a design rule for robust phase slopes under a time-varying  dominant path can be derived. 

A summary of the contributions of this paper follows.
\begin{itemize}
	\item Based on a reference highway scenario, that takes into account distance, speed, and antenna separation,  we model the time variation of phase shifts between antennas, \gls{PL}, and \gls{AOA}-dependent antenna responses using an affine function at medium to large distances between the \gls{Tx} and \gls{Rx}.
	\item We analytically derive the loss function in \gls{ACN} sum-\gls{SNR} when the three quantities vary separately.
	\item We set a design rule to choose the most robust phase slope against the effects of time variation of the single dominant component quantities.
	\item Using numerical computation we validate the design rule and illustrate the performance of \gls{ACN} under time-varying \gls{AOA}, phase shifts, and \gls{PL}. 
\end{itemize}

\section{Preliminaries}
In this section, we briefly reintroduce \gls{ACN}/\gls{ABN}, and restate the worst-case propagation assumptions used to design them alongside the obtained optimal phase slopes in~\cite{ACN,lehocine2021abn}.

\subsection{Multiple Antenna Scheme}
Consider the multiple antenna scheme shown in \fig~\ref{Fig:ACN} with analog time-varying phase shifters that are modeled following 
\begin{align}\label{eq:phaseShifter:ACN}
 \mathrm{e}^{\jmath(\ph_l t+ \phI_l)}, \quad 0\leq l\leq L-1,
\end{align}
where $\ph_l\in \mathbb{R}$, and $\phI_m\in [0,2\pi)$ denote the phase slope and the initial unknown phase offset, respectively, and $L$ is the number of antennas. The scheme is referred to as \gls{ACN} when used at the \gls{Rx}, and \gls{ABN} when used at the \gls{Tx}.
A multiplier of $1/\sqrt{L}$ is introduced in~\eqref{eq:phaseShifter:ACN} to ensure equal transmitted power with respect to a single transmit antenna case when the scheme is used at the transmitter.


\subsection{Data Traffic Model and Performance Metric}
Consider a periodic traffic of \glspl{CAM} that are broadcasted by \glspl{VU} every $0.1\leq T\leq 1~$s~\cite{CAM}. \glspl{CAM} carry status information like position, speed, heading, etc. Their repetition interval $T$ depends on how fast the dynamics of a vehicle are changing, the channel load, and the requirements of \gls{C-ITS} applications~\cite{CAM}. \glspl{CAM} are short packets with sizes that range between $100$ and $500$ bytes~\cite{IEEE11p_rate_Jiang2008}, which correspond to a duration $\Tm< 0.7$~ms, assuming IEEE802.11p as an access technology and $6$~Mbit/s data rate. (For LTE-V2X technology, the \gls{CAM} duration corresponds to a subframe duration $\Tm = 1$~ms~\cite{Molina2020_11pVsLTE-V2X}.)
Observe that the \gls{CAM} duration is much smaller than the repetition interval $T$, $\Tm\ll T$, and this will be used to assert certain assumptions in the sections to come.

Consider that \gls{AoI}~\cite{Aoinfo} is used to assess the reliability of a \gls{C-ITS} application that relies on the information carried by \glspl{CAM}. At a receiving \gls{VU}, if the information available about a certain neighboring vehicle has not been updated within a maximum tolerable \gls{AoI}, $A_{\textrm{max}}$, then the \gls{C-ITS} application running at the receiving \gls{VU} is in outage. If latency between the transmission and reception of packets is neglected, an outage occurs if a burst of $K$ consecutive \glspl{CAM} is lost, implying that $A_{\textrm{max}}= KT$. The loss of $K$ consecutive packets was defined as \gls{BrEP} in~\cite{ACN}, and used in~\cite{ACN,lehocine2021abn} as a design metric for \gls{ACN} and \gls{ABN}. In particular, the multiple antenna schemes were designed to minimize the \gls{BrEP}. Under the assumption\footnote{Motivation and justification can be found in~\cite[\secR~III]{lehocine2021abn}.} that packet error probability follows an exponential function of \gls{SNR}, and that packet errors are statistically independent, minimizing \gls{BrEP} is found to be equivalent to maximizing the sum of \gls{SNR} of the $K$ consecutive packets~\cite[\secR~III]{ACN},~\cite[\secR~III.B]{lehocine2020HC}. Therefore, the sum-\gls{SNR} is the metric used in assessing the \gls{ACN}/\gls{ABN} performance in this paper. 

\begin{figure}[t]
	\centering
	\includegraphics[width=0.8\columnwidth]{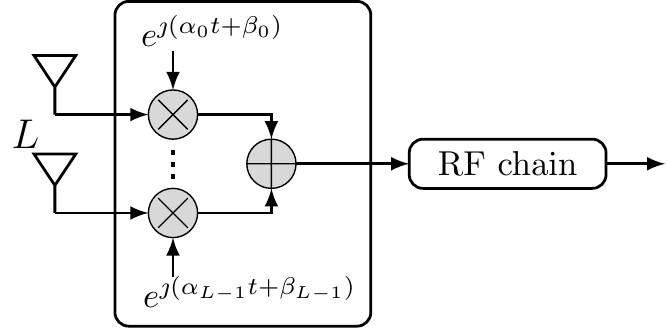}
\caption{\gls{ACN} with $L$ receive antennas.}
	\label{Fig:ACN}
\end{figure}
	\subsection{Channel Model and Worst-Case Propagation Scenario}\label{sec:sub:assumptions}
To ensure robust communication, \gls{ACN}/\gls{ABN} was designed under a scarce multipath propagation with a dominant path, and a few diffuse components with small angular spread, which is  typical in roads that are not surrounded by buildings, e.g., highways~\cite{channel_v2v_angular}. Such propagation is challenging for the antenna system since if the \gls{AOD} and the \gls{AOA} of the dominant path coincide with a direction where the antenna system has low gain, there is a risk to lose a packet. Furthermore, in the cases when the \gls{AOD} and the \gls{AOA} are approximately non-varying over the duration of $K$ consecutive packets, then there is a risk of an outage, i.e., loss of $K$ consecutive packets. Following this, the baseband channel between the transmitter and the receiver was modeled in~\cite{ACN,lehocine2021abn} using the single dominant physical path. Assuming a $1\times L$ \gls{ACN} system (this is equivalent to $L\times 1$ \gls{ABN}), 
the channel gain is given by~\cite[Ch.~6]{WirelessCom_Molisch}
	\begin{align}
h_{l}(t)= a(t) g_l(\phi)\mathrm{e}^{-\jmath 	\om_l(t)}, \quad  l=0,\ldots,L-1,\label{def:channel:v1} 
\end{align}
where $a(t)$ is the (complex) path amplitude, $\phi$ is the \gls{AOA}, $g_l$ is the far-field function of the receive antenna $l$ in the azimuth plane, and $\om_l$ is the relative phase shift at antenna $l$ with respect to the reference antenna with index $l=0$. It is given by 
\begin{align}\label{eq:om:d}
	\om_l(t)=2\pi/\lambda (d_{l}(t)-d_{0}(t)),
\end{align}
where $d_{l}(t)$ is the propagation distance between the receive antenna with index $l$ and the transmit antenna, and $\lambda$ is the carrier wavelength. Note that $\om_0(t)=0$. The far-field function of the transmit antenna, which can be accounted for by $a(t)$, is assumed for simplicity to be isotropic throughout the paper, and thus $a(t)$ models solely the amplitude of the dominant path. 

Since a dominant path with a negligibly varying $\phi$ results in a higher risk of an outage, in~\cite{ACN,lehocine2021abn} it is assumed that the speed and position of the \gls{Tx} and the \gls{Rx}, alongside potential interacting objects, are such that 
$\phi$ is approximately constant over the duration of $KT$~s. Consequently, $\om_l$ and \gls{PL} are also assumed to be approximately constant over $KT~$s. These are referred to as the worst-case propagation assumptions, and they were used when designing \gls{ACN}/\gls{ABN}.



\subsection{Sum-SNR Under Worst-Case Propagation}
Given the transmission of a \gls{CAM} packet, employing~\eqref{eq:phaseShifter:ACN} and ~\eqref{def:channel:v1}, the received signal by a \gls{VU} after \gls{ACN} combining (as shown in \fig~\ref{Fig:ACN}) is expressed as
\begin{align}\label{eq:r:signal}
	r(t) = a(t)x(t)\sum_{l=0}^{L-1} g_l(\phi) \mathrm{e}^{-\jmath (\om_l -\ph_l t - \phI_l)} + \sum_{l=0}^{L-1}n_l(t),
\end{align}
where $x(t)= \tilde{x}(t-\tau(t))$, $\tilde{x}(t)$ is the transmitted baseband signal, $\tau(t)=2\pi d_0(t)/\lambda$ is the propagation delay, and $n_l(t)$ is an independent zero-mean additive white Gaussian noise over the signal bandwidth with variance $\sigma_{\textrm{n}}^2$.
Let $P_{\textrm{r}}=\mathbb{E}\{|a(t)x(t)|^2\}$ be the average received power which is approximately the same for the $K$ packets since the \gls{PL} is assumed to be approximately constant over $KT~$s, and let
\begin{align}
\psi_l(t)&\triangleq \om_l(t)-\phase{g_l(\phi)}-\phI_l,
\label{def:psi}
\end{align}  
which is non-varying over $KT$~s since $\om_l(t)$ is assumed to be non-varying over the same duration.
Then, using~\eqref{eq:r:signal} and~\eqref{def:psi}, we express the \gls{SNR} of the $k^{\text{th}}$ packet as
\begin{align}\label{eq:SNRk}
	\gamma_k = \frac{P_{\mathrm{r}}}{L\sigma^2_{\mathrm{n}}} \bigg|\sum_{l=0}^{L-1} |g_l(\phi)|\mathrm{e}^{-\jmath(\ps_l-\ph_l kT)} \bigg|^2,
\end{align}
where the phase variation is assumed negligible over a packet duration since $\Tm \ll T$, and thus the approximation $\mathrm{e}^{-\jmath(\ps_l-\ph_l t)}\approx \mathrm{e}^{-\jmath(\ps_l-\ph_lkT)}$ when $kT\leq t\leq kT+\Tm$, is employed when deriving $\gamma_k$.
Then we can readily express the normalized sum-\gls{SNR} with respect to $P_{\mathrm{r}}/\sigma^2_{\mathrm{n}}$ using the column vectors $\phV=[\ph_0,\ph_1,\ldots,\ph_{L-1} ]^\textsf{T}$ and $\psV=[ \ps_0,\ps_1,\ldots,\ps_{L-1} ]^\textsf{T}$, as
\begin{align}
S(\phi,\phV,\psV )&=\sigma^2_{\mathrm{n}}/P_{\mathrm{r}}\sum_{k=0}^{K-1}\gamma_k\nonumber\\ 
&=K\sum_{l=0}^{L-1}\frac{|g_l(\phi)|^2}{L}+
J(\phi,\phV,\psV ),
\label{eq:S:GplusJ}
\end{align}
where 
\begin{align}
J(\phi,&\phV,\psV )=\frac{2}{L} \sum_{l=0}^{L-2}\sum_{m=l+1}^{L-1}|g_l(\phi)||g_m(\phi)|  \nonumber \\
&\times\sum_{k=0}^{K-1}\cos\big(\ps_m-\ps_l-(\ph_m-\ph_l)kT\big). \label{eq:J}
\end{align}

To design \gls{ACN}/\gls{ABN} an optimization problem is defined in~\cite{ACN,lehocine2021abn} to find the phase slopes that maximize the sum-\gls{SNR} for the worst-case \gls{AOA}/\gls{AOD}, and for the worst-case $\psV$ since it depends on the initial unknown and uncontrollable phase offsets $\{\phI_l\}$. 
The solution to the optimization problem is found in~\cite{ACN,lehocine2021abn} to be the same for any \gls{AOA}/\gls{AOD}, and it is given by
	\begin{align}\label{eq:optim}
S^\star(\phi)&=\sup_{\phV\in\mathbb{R}^L}~ \inf_{\psV\in [0,2\pi]^L }S(\phi,\phV,\psV )=K\sum_{l=0}^{L-1}\frac{|g_l(\phi)|^2}{L}.
\end{align}
The optimal phase slopes are independent of the far-field functions of the antennas and exist when $L \leq K$ and they satisfy~\cite{ACN,lehocine2021abn} 
\begin{align}\label{eq:optimality:condition}
		(\ph_m-\ph_l)T/2 \in \mathcal{X^*}, \quad    0\leq l<m\leq L-1,
\end{align}
where
\begin{align}\label{def:Xstar}
	\mathcal{X^*}\triangleq\{q\pi/K:q\in\mathbb{Z}\}\setminus\mathcal{X},\quad \mathcal{X} \triangleq \{q\pi: q\in \mathbb{Z} \}.
\end{align}

The condition in~\eqref{eq:optimality:condition} can be satisfied when $L=2$ using a single phase shifters, where $\ph_0=\phI_0= 0$, and $\ph_1 T/2\in  \mathcal{X^*}$, $\phI_1\in [0,2\pi)$. 


	\begin{figure}[b]
	\includegraphics[width=\columnwidth]{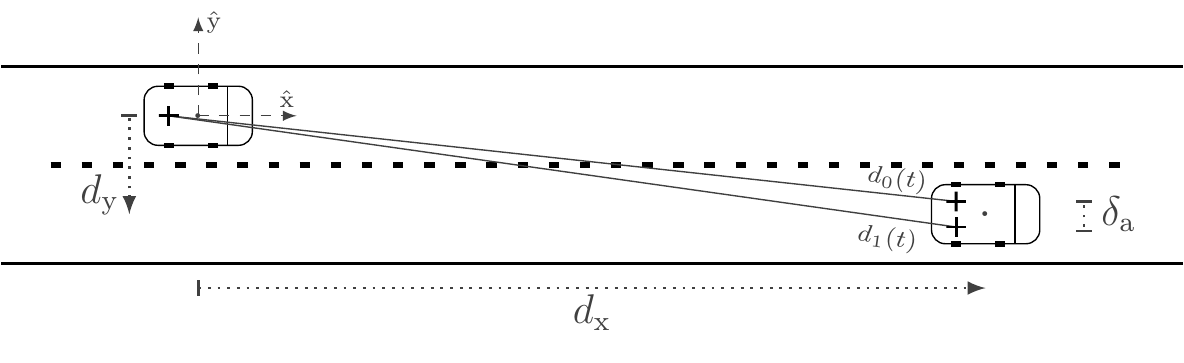}
	\caption{Two-lane highway scenario with two reference vehicles.}
	\label{fig:scenario}
\end{figure}

\section{Time-Varying Single Dominant Path Propagation}
The optimal phase slopes for \gls{ACN}/\gls{ABN} were designed under the worst-case assumption of a negligibly varying single dominant path over $KT~$s. In the following, we aim to investigate the effects of time variation of the dominant path on \gls{ACN}/\gls{ABN} performance. That is done following an investigation of the impact of variation of $\phi$, $\om_l$, and \gls{PL} separately. Moreover, we aim to define a design rule to pick the most robust choice of phase slopes to the variation of these quantities. 
To be able to perform an analytical investigation of the variation of $\phi$, $\om_l$, and \gls{PL} we resort to the simple scheme of a $1\times2$ system. It allows us to accurately understand the effects of the quantities and to draw design methodology. These can be used as guidelines when investigating the time variation effects of these quantities on $\Lt\times \Lr$ \gls{ABN}-\gls{ACN} system.


\begin{figure}
	\includegraphics[width= \columnwidth]{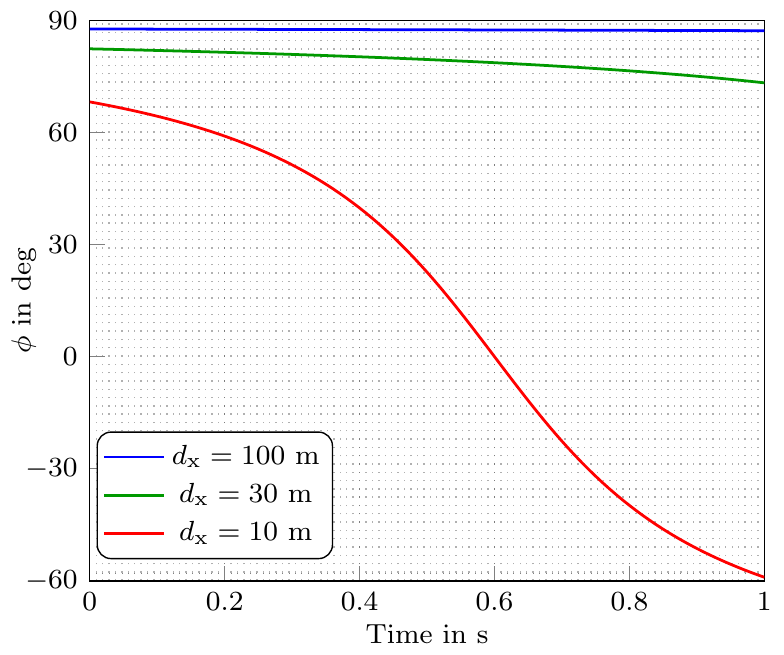}
	\caption{Time variation of $\phi$ for different initial distances between the \gls{Tx} and the \gls{Rx}, $\dv = -60~$km/h, $\Lw= -4~$m.}
	\label{fig:phi:var}
\end{figure} 
\begin{figure}
			\includegraphics[width= \columnwidth]{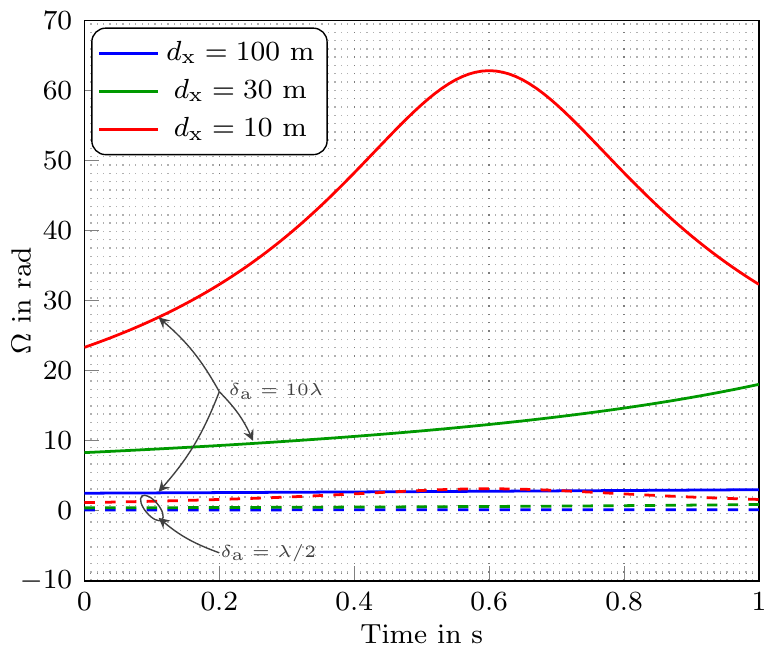}
	\caption{Time variation of $\om_1$ for different initial distances between the \gls{Tx} and the \gls{Rx}, $\dv = -60~$km/h, $\Lw= -4~$m.}
	\label{fig:Omega:var}
\end{figure}

To model the variation of $\phi$, $\om_l$, and \gls{PL} and to evaluate the sum-\gls{SNR} of the multiple-antenna scheme, we consider two reference \glspl{VU} on a highway as shown in \fig~\ref{fig:scenario}, one employing two antennas $L=2$ and it is referred to as the receiving \gls{VU}, and the other is employing a single antenna, and it is referred to as the transmitting \gls{VU}. We assume a \gls{LOS} propagation between the two \glspl{VU}, which corresponds to the dominant component.
 If the two \gls{VU} are moving at the same speed, then the worst-case propagation assumptions in \secR~\ref{sec:sub:assumptions} (non-variation of the channel) are fully satisfied. In addition, if the two vehicles are moving at different speeds but in the same lane, then the \gls{AOA} and $\om_l$ are fixed, while the \gls{PL} varies. However, if the vehicles are moving at different speeds and located on different lanes, then the geometries between the \gls{Tx} and \gls{Rx} change over time and result in variations of all three quantities of interest. In other words, the accuracy of the aforementioned worst-case scenario assumptions depends on the speed and lane position of vehicles (for this particular scenario).
The use of such a scenario allows us to model the variations of the three quantities of interest. 

To define certain parameters for the scenario in \fig~\ref{fig:scenario} we take the initial position of the transmitting \gls{VU} as the origin $(0,0)$ of a Cartesian coordinates system. Then, we define the following parameters.
\begin{itemize}
	\item $\dx$, the initial longitudinal position of the receiving \gls{VU} with respect to the transmitting \gls{VU} (measured from the center of vehicles). 
	\item $\Lw$, the lateral position of the receiving \gls{VU} with respect to the transmitting \gls{VU}. 
	
	\item $\dv=(v_{\textrm{r}}-v_{\textrm{t}})$, the speed difference between the receiving and the transmitting \glspl{VU}.
	\item $\da$, the receive antennas separation. 
\end{itemize}
For simplicity, we assume that the antennas of the \gls{Rx} are mounted along the lateral axis of the vehicle. 
\begin{figure}
	\includegraphics[width= \columnwidth]{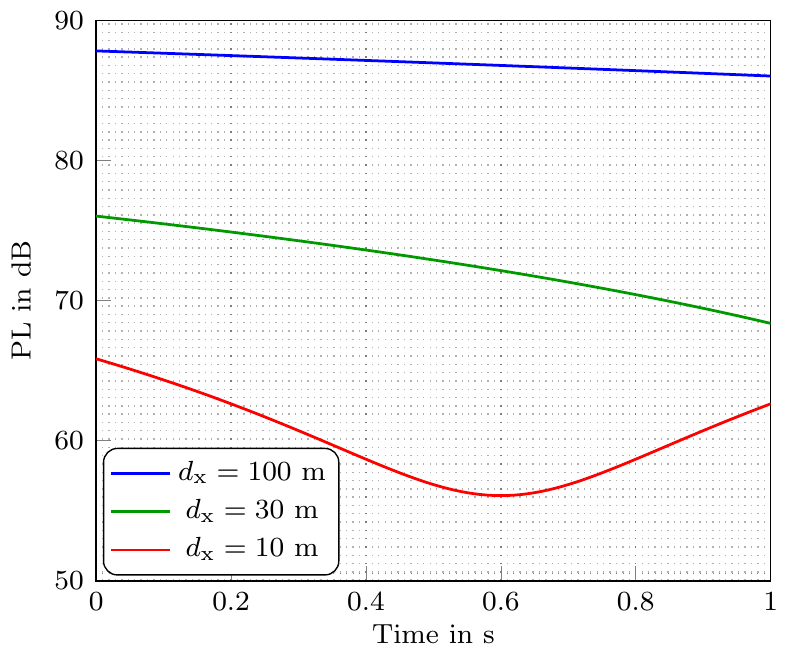}
	\caption{Time variation of \gls{PL} following WINNER+B1 pathloss model~\cite{3GPP_rel14_36p885} for different initial distances between the \gls{Tx} and the \gls{Rx}, $\dv = -60~$km/h, $\Lw= -4~$m.}
	\label{fig:PL:var}
\end{figure}

The defined parameters $\dx$, $\Lw$, $\dv$ and $\da$ can span a wide range of values. However, we limit the range of these parameters to intervals that are large enough to allow us to observe the different effects resulting from changing geometries between the \gls{Tx} and \gls{Rx}.  In particular, we set $\dx\in [-100,100]~$m, since for larger distances than $100$~m, $\phi$, $\om_l$ and \gls{PL} are approximately non-varying over $KT$~s.  From another aspect, taking into account typical cruising speeds on highways, and assuming a maximum regulatory speed of $100~$km/h, we can restrict $\dv= (v_{\textrm{t}}-v_{\textrm{r}}) \in [-60,60]~$km/h, which corresponds to absolute speeds in the range $[70,130]~$km/h. 
Furthermore, to cover different lane positions we use $\Lw\in[-4,0,4]$ corresponding to a lane width of $4~$m.
Lastly, we consider two cases for $\da$,  $\da=\lambda/2\approx 0.025~$m and $\da=10\lambda\approx 0.5~$m at carrier frequency $f_{\textrm{c}}=5.9~$GHz , corresponding to a small and large antenna separations, respectively. 

In Figs.~\ref{fig:phi:var}--\ref{fig:PL:var} we show examples of the time variation of the three quantities 
for different $\dx$ when $\Lw=-4$~m, and $\dv=-60$~km/h, which is the speed within $[-60,60]$~km/h that results in the maximum change of the quantities.
From the figures, we observe that at $\dx = 100~$m the \gls{AOA}  exhibits negligible change and so does $\om_1$ ($\om_0=0$) and the \gls{PL}, implying that the worst-case assumptions in \secR~\ref{sec:sub:assumptions} hold.
At a short distance $\dx=10~$m, $\phi$ is very rapidly changing. Thus, even if the \gls{AOA} for a received packet coincides with a low gain of the antenna system, the \gls{AOA} for the successive packets (e.g., when $T=0.1$~s) can be coinciding with a high gain of the antenna system, which lowers the risk of consecutive packet losses. 
More importantly, the \gls{PL} is very low at such short distances as seen in \fig~\ref{fig:PL:var} which results in a high \gls{SNR} and a high probability of decoding packets. Thus, this region is not critical for the communication system. 


At the medium distance $\dx = 30~$m, the \gls{AOA} is slowly changing, with an approximate rate of $9~\deg/1~$s. 
 In this case, if the \gls{AOA} for a received packet coincides with a low gain of the antenna system, there is a risk that several successive packets (e.g., when $T=0.1$~s, $\phi$ changes by $0.9$~deg/$T$) experience comparably low gain as well. 
The \gls{PL}, which is also slowly changing, is neither too high nor too low in this region. 
Therefore, there is a risk of losing a burst of consecutive packets due to the alignment of the dominant path with a direction of low gain of the antenna system. From another aspect, the change in $\om_1$ is fast (greater than $2\pi$~rad$/1$~s for $\da=10\lambda$), which impacts the \gls{ACN}/\gls{ABN} system by introducing a time-varying phase offset to the preset phase shifters. Proper antenna processing is of paramount importance in this region of medium to large distances, and it will therefore be the focus of the coming analytical investigations. 

In the coming three sections, we investigate the effects of time variation of $\phi$, $\om_l$, and \gls{PL} on \gls{ACN}/\gls{ABN} performance.
When $\phi$ changes, there is a change in the channel phase shift which is captured by $\om_l$, and a change in antenna far-field functions $g_l(\phi)$, in both phase and amplitude. When we refer to time variation effects of $\phi$ we mean the far-field functions variation effects. 
To simplify the analysis, we study the effects of time variation of each quantity separately. We start with $\om_l$, followed by \gls{PL}, and finally $\phi$.
As mentioned earlier, the focus of the analysis is on medium to large distances. This will be backed up by a numerical assessment of the system performance at short distances in the numerical results section.


\section{Channel Phase ($\Omega$) Variation Effects}\label{sec:omega}
In this section, we first model $\om_l$ based on the reference scenario shown in \fig~\ref{fig:scenario} when the distance between the two \glspl{VU} is medium to large. We then derive the sum-\gls{SNR} of $1\times 2$ \gls{ACN} system under the effects of time variation of $\om_l$. Last, we derive a design rule to pick a robust phase slope under these conditions.






\subsection{Channel Phase Shift ($\om$) Model}\label{sec:sub:omega:var}
Given that $L=2$, and $\Omega_0(t)=0$, under \gls{LOS} propagation between the reference \glspl{VU} shown in \fig~\ref{fig:scenario}, the relative phase shift at receive antenna $l=1$ with respect to the reference antenna ($l=0$), $\om_1(t)\triangleq \om(t)$ is modeled by 
\begin{align}
\om(t)	&= 2\pi/\lambda (d_{1}(t)-d_{0}(t)) \nonumber \\
				& = 2\pi/\lambda \bigg(   \sqrt{(\dx+\dv   t)^2+(\Lw-\da/2)^2}  - 			\nonumber\\
				& \qquad \sqrt{(\dx+\dv   t)^2+(\Lw+\da/2)^2}      \bigg).\label{eq:omega:full}
\end{align}
Observe that as $\dx$ becomes large $\dv /\dx \rightarrow 0$, and $\om$ becomes approximately constant as seen in \fig~\ref{fig:Omega:var}.  Moreover, the smaller the antenna separation is $\da$, the slower the time variation of $\om$.
To investigate the effect of this time variation, we resort to the approximation
\begin{align}
		\om(t) \approx  \bom+ \aom~t,\label{eq:omega:app:affine}
\end{align}
where $\aom $ and $\bom$ can be obtained using first-order Taylor series, or the \gls{LS} method. 
When first-order Taylor expansion around $t=t_0$ is used, $\aom$ is given by
\begin{align}
&\aom= \frac{2\pi}{\lambda} \dv    \bigg (\frac{(\dx+\dv t_0)}{\sqrt{(\dx+\dv t_0) ^2+(\Lw-\da/2)^2}} \nonumber \\ &\qquad -\frac{(\dx+\dv t_0)}{\sqrt{(\dx+\dv t_0)^2+(\Lw+\da/2)^2}} \bigg),\label{eq:a:omega:taylor}
\end{align}
and $\bom = \om(t_0) -\aom t_0$. 
Based on \fig~\ref{fig:Omega:var}, we expect the affine approximation to have good accuracy at medium to large distances with $\dv \in [-60,60]~$km/h. At short distances, the approximation is not expected to be very accurate, especially when the antenna separation is large (e.g., $\da = 10 \lambda$), and the speed difference is high.

\subsection{Sum-\gls{SNR}}
To evaluate the sum-\gls{SNR} of \gls{ACN}/\gls{ABN} taking into account the effects of time variation of $\om$ only, we assume the use of isotropic antennas at the \gls{Rx}, i.e., $g_0(\phi) = g_1(\phi) = 1$, $\forall \phi$. Thus, despite that $\phi$ varies, the gain of the antenna system does not vary. Moreover, we assume that \gls{PL} is approximately non-varying over $KT$~s. This allows us to disentangle the effects of variations of the channel phase $\om$ from the effects of variation of \gls{AOA} and \gls{PL}.
Now, assume that the \gls{ACN} is implemented using a single phase shifter, implying that $\ph_0=\phI_0=0$, and $\ph_1 \triangleq \ph \in  [0,2\pi) $, $\phI_1\triangleq \phI\in [0,2\pi)$. Then, from~\eqref{def:psi} and~\eqref{eq:omega:app:affine}, we have 
\begin{align}
	\ps_1(t)&= \om(t)-\phase{g_1(\phi)}-\phI\nonumber \\
	&\approx \underbrace{( \bom -\phase{g_1(\phi)} -\phI )}_{\psiZ} + \aom ~t, \label{eq:omega:var} 
\end{align}
while $\ps_0 = -\phase{g_0(\phi)}$ ($= 0$ when the antenna is isotropic). 
 Assuming that the variation of $\ps_1(t)$ due to $\aom  $ is negligible over a packet duration $\Tm\ll T$, we write
\begin{align}\label{eq:psi:packet:approx}
\ps_1(t)\approx \ps_1(kT),\qquad kT\leq t\leq kT+\Tm .
\end{align}

Following this, the sum-\gls{SNR}~\eqref{eq:S:GplusJ} can be expressed taking into account~\eqref{eq:omega:var}, the isotropic antenna characteristics, and $L=2$ as
	\begin{align}
	\SsnrA(\phi,\phV,\psV )&=\sigma^2_{\mathrm{n}}/P_{\mathrm{r}}\sum_{k=0}^{K-1}\gamma_k
	=K +\JA(\phi,\phV,\psV ),\label{eq:S:iso:1st}
	\end{align}
where $\JA$ is given by~\eqref{eq:J} and it simplifies to 
\begin{align}
	\JA(\phi,\phV,&\psV  )= |g_0(\phi)||g_1(\phi)|  \nonumber \\ \times&\sum_{k=0}^{K-1}\cos\big(\ps_1(kT)-\ps_0 -(\ph_1-\ph_0)kT\big)\nonumber \\
	&\approx \sum_{k=0}^{K-1} \cos\big(\psiZ + \aom kT- \ps_0-\ph kT\big)\nonumber \\
	&=\sum_{k=0}^{K-1} \cos\big((\psiZ-\ps_0) -(\ph-\aom  )kT\big)\nonumber \\
	&=\sum_{k=0}^{K-1} \cos\big(y -2 x k\big)\nonumber \\
	&\triangleq \JA(x, y), \label{eq:J:omega}
\end{align}
where the approximation follows from~\eqref{eq:omega:var}, and where
\begin{align}\label{def:y:x}
	y \triangleq (\psiZ-\ps_0)\quad ,\quad x \triangleq (\ph-\aom  ) T/2.
\end{align}
We can write the sum-\gls{SNR} as
\begin{align}\label{eq:S:iso}
	\SsnrA(x, y )= K + \JA(x, y).
\end{align}

Now, we can evaluate the performance of the system using a similar optimization problem as in~\eqref{eq:optim},
\begin{align}\label{eq:S:min}
	\inf_{y \in[0,2\pi) }\SsnrA(x, y) = K \bigg(1 + \inf_{y \in[0,2\pi) } \frac{\JA(x, y )}{K} \bigg).
\end{align}
That is, we account for the worst-case value of $y$, since it depends on the unknown, uncontrollable, initial phase shift  $\phI$ of the \gls{ACN}.  
When the channel phase variations are negligible, $\aom =0$, we know from~\eqref{eq:optim} that the optimal performance is given by $\inf_{y  }\SsnrA(x, y) = K$, which implies that $\inf_{y} \JA(x, y ) = 0$.
Therefore, we define the loss function as 
\begin{align}\label{eq:loss:omega}
\LA(x)\triangleq -\inf_{y \in[0,2\pi) } \frac{\JA(x, y )}{K},
\end{align}
and we characterize some of its properties in the following lemma.
\begin{lemma}\label{lemma:loss}
	Let $\LA(x)$ be as defined in~\eqref{eq:loss:omega}, $x\in \mathbb{R}$, and let $\mathcal{X}$ be as defined in~\eqref{def:Xstar}, then
	\begin{enumerate}
	    \item The loss function is given by
	    \begin{align}\label{eq:lemma:loss:cases}
		\LA(x)=\displaystyle \begin{cases}
		1, & x\in \mathcal{X}\\
	\displaystyle |f_1(x)|/K,& x\notin \mathcal{X}
		\end{cases}
	\end{align}
	where $f_1: \mathbb{R}\setminus{\mathcal{X} \rightarrow \mathbb{R}}$, is given by
	\begin{align}
        f_1 (x) =   \frac{\sin(Kx)}{\sin(x)}  \label{eq:f1}.
        \end{align}
	\item  $\LA$ is periodic with period $\pi$, and symmetric around $\pi/2$.
	
	\item The loss function is bounded as
	\begin{align}\label{eq:lemma:loss:LB}
		0\leq	\LA(x) \leq 1 ,\quad x\in [0,\pi),
		\end{align}
		where $\LA(x)=0$ when $x\in \mathcal{X}^\star$, and $\mathcal{X}^\star$ is defined in~\eqref{def:Xstar}.
	    
	\end{enumerate}
	\begin{proof}
		See the Appendix.
	\end{proof}
	\end{lemma}
\begin{figure}
	\includegraphics[width= \columnwidth]{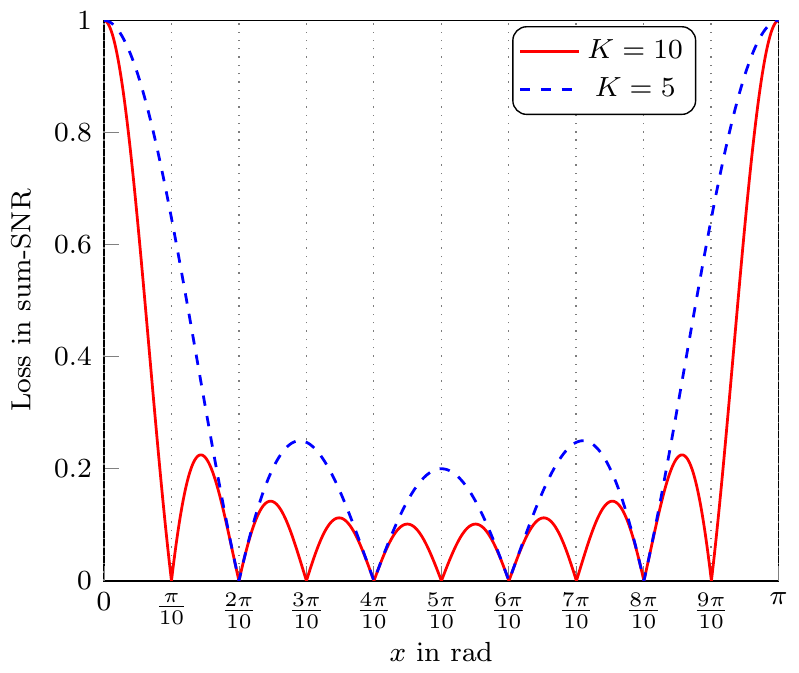}
	\caption{The sum-\gls{SNR} loss function $\LA(x)$, $x= (\ph -\aom )T/2$ when $K=5,10$. }
	\label{fig:L}
\end{figure}
\begin{figure}
	\centering
	\includegraphics[width= \columnwidth,angle=0]{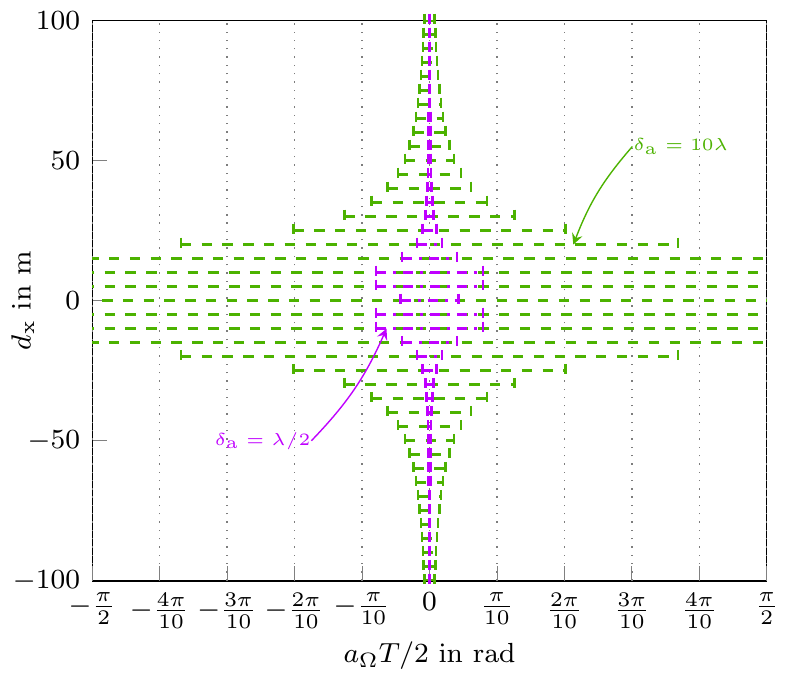}
	\caption{Range of change in $\om$ as a function of $\dx$ when $\dv\in[-60,60]~$km/h, and $\Lw=\pm4$. First order Taylor expansion~\eqref{eq:a:omega:taylor} is used to compute $\aom$.}
	\label{fig:a_om:range}
\end{figure}

The loss function for $K=5$ and $K=10$ is shown in \fig~\ref{fig:L}. 
Since $\LA(x)$ is periodic with period $\pi$, we define
\begin{align}\label{def:Astar}
\Astar\triangleq \bigg\{ q \frac{2\pi}{KT}; q=1,2,\ldots, K-1 \bigg \},
\end{align}
which has $K-1$ elements satisfying $\ph \in \Astar$, $\ph T/2\in \Xstar$. These are all the unique elements in $\Xstar$, since for every $x \in \Xstar$, we can write $x = \ph T/2 + q\pi$, where $\ph \in \Astar$, $q\in \mathbb{Z}$.

Recall that when $\aom =0$, all $\ph \in \Astar$ yield the same optimal sum-\gls{SNR}, and $\LA(x) = 0$.
Consider now an example where $\aom T/2 \in [-\pi/K,\pi/K]$, and $\ph T/2 = \pi/K$ ($\ph \in \Astar$). It follows that $x = (\ph -\aom)T/2 \in [0, 2\pi/K]$, and thus, the worst-case loss due to phase deviation $\aom$ correspond to $\LA(x)=1$, as can be seen from \fig~\ref{fig:L}. On the other hand, when $\ph T/2 = 2\pi/K$, ($\ph \in \Astar$) it follows that $x \in [\pi/K, 3\pi/K]$, and from \fig~\ref{fig:L} we see that the worst-case loss in this case is $\LA(x)\approx 0.25$, and $\LA(x)\approx 0.22$ for $K=5$, and $K=10$, respectively.  
Hence, under channel phase variation, not all phase slopes $\ph \in \Astar$ yield the same loss in sum-\gls{SNR}. Moreover, $\ph T/2 = 2\pi/K$ is more robust than $\ph T/2 = \pi/K$ in mitigating phase deviation effects when $\aom T/2 \in [-\pi/K,\pi/K]$. Therefore, in the following, we propose a design rule to choose the most robust phase slope in $\Astar$ under channel phase shift variation.

\subsection{Phase Slope Design Under Time-Varying $\om$ ($\aom \neq 0$)}
To derive a design rule to choose the most robust phase slope in $\Astar$ under channel phase variation, we start by characterizing the trends of the loss function, $\LA(x)$.
Let $\ph \in \Astar$, and assume that $\aom  \neq 0$. At the points $x=(\ph-\aom )T/2=q\pi $ the loss function is at its maximum value $\LA(x)=1$.
On the other hand, as can be seen in \fig~\ref{fig:L}, in the range $[\pi/K, (K-1)\pi/K ]$ the loss function exhibits $K-2$ lobs. These lobes have decreasing maxima in the range $[\pi/K,\pi/2]$.
That follows since $\sin(x)$ is increasing over the interval, while $|\sin(Kx)|$ is periodic with $\pi/K$ (which is equal to the lobes width). Due to the symmetry around $\pi/2$, the lobes maxima are increasing as we get far from the symmetry point in the range $[\pi/2,(K-1)\pi/K]$.
In addition, the maximum loss within $[\pi/K, (K-1)\pi/K]$ is bounded following
\begin{align}
\LA(x) \leq \frac{1}{K\sin(\pi/K )}, \quad   	x\in \bigg[\frac{\pi}{K},(K-1)\frac{\pi}{K} \bigg],
\end{align}
since $|\sin(Kx)|\leq 1$ and $ \sin(\pi/K) \leq \sin(x) $ where $x \in [\pi/K, (K-1)\pi/K]$. 
This bound is decreasing with $K$ and it quantifies to $\LA(x) \leq 0.34$ for $K=5$, and $\LA(x) \leq 0.32$ for $K=10$.
Substituting with the bound in~\eqref{eq:S:min} we deduce that the sum-\gls{SNR} when $x \in [\pi/K, (K-1)\pi/K ]$ is at worst $-1.8$~dB, and $-1.7$~dB below the zero-loss sum-\gls{SNR} ($\inf_y \SsnrA(x,y) = K$) for $K = 5$, and $K = 10$, respectively.
Hence, in summary, the loss in sum-\gls{SNR} is moderate 
within the interval $[\pi/K, (K-1)\pi/K]$ and it gets lower as we approach the point of symmetry $\pi/2$ of the loss function. On the other hand, the loss in sum-\gls{SNR} is severe in the intervals centered around the points $(\ph-\aom )T/2=q\pi $.

Guided by these trends in the loss function $\LA(x)$, we aim to pick the phase slope $\ph \in \Astar$ that handles a wide range of error due to channel variation $\aom $ without resulting in the most severe loss in sum-\gls{SNR} occurring at $x= (\ph- \aom )T/2 =q\pi$. 
This is equivalent to picking the phase slope that has the largest phase distance from the points with the most severe loss, $x=q\pi$, and which is in turn, the phase slope that has the shortest phase distance from the point $x=\pi/2$. We formally define this as  
\begin{align}\label{eq:design:rule}
	\alpha^\star \triangleq\arg \min_{\ph \in \Astar} |\ph T/2 -\pi/2|.
\end{align}
Besides having the largest phase distance from the points with the most severe loss, the \gls{ACN}  with $\ph^\star$ ensures that the effective phase slope $(\ph^\star - \aom )T/2$ is within the region with moderate loss in sum-\gls{SNR}, $ [\pi/K, (K-1)\pi/K ]$ for the widest range of $\aom $ compared to any other phase slope in $\Astar$. Thus, $\ph^\star$ is the most robust phase slope in $\Astar$ when $\Omega$ is time-varying.
The solutions to~\eqref{eq:design:rule} are given by
\begin{align}\label{eq:alpha:star}
\alpha^\star \in  \begin{cases}
\bigg \{  \displaystyle \frac{K}{2}\frac{2\pi}{KT}  \bigg \} , & K \text{ even }\\
\bigg \{ \displaystyle  \frac{K-1}{2}\frac{2\pi}{KT} ,  \displaystyle \frac{K+1}{2}\frac{2\pi}{KT} \bigg \} ,& K \text{ odd }
\end{cases}
\end{align}

To get more insight into the design rule~\eqref{eq:design:rule} we plot in \fig~\ref{fig:a_om:range} the range of $\aom $ as function of $\dx$ when $\dv \in [-60,60]~$km/h. 
\textbf{(i)} We observe that for $\da=\lambda/2$ the variation of $\aom $ is very limited at medium to large distances. This indicates that the impact of channel phase variation is not severe when the antenna separation is small. 
\textbf{(ii)} For $\da=10\lambda$, we see that $\aom $ has much wider range compared to $\da=\lambda/2$. The interval of $\aom$ increases with the increase of speed difference $|\dv|$ or the decrease of distance $|\dx|$. 
This is an indication that certain choices $\ph\in \Astar$ can have better properties in mitigating the effects of variation of $\Omega $ over a wide range of speed and distances. This supports the proposed design rule in~\eqref{eq:design:rule}. 
For example, assuming $K =10$, then $\alpha^\star T/2=\pi/2$. Following that, from~\fig~\ref{fig:a_om:range} we obtain that when $|\dx|\geq20~$m, $|\aom| T/2< \pi/2$ for $\dv\in[-60,60]$. Thus, when $|\dx|\geq20$, the effective phase slope $(\alpha^\star - \aom )$ satisfies $0<(\alpha^\star - \aom )T/2 <\pi$, implying that $\alpha^\star$ allows us to avoid severe loss in sum-\gls{SNR} for any speed when $|\dx|\geq20~$m.  
On the other hand, using $\ph T/2 =\pi/K$, then $0< (\ph-\aom )T/2 <\pi$ is satisfied for any $\dv\in[-60,60]$, only when $|\dx|\geq 35~$m. Hence, for medium distances below $35~$m, there exists $\dv\in[-60,60]$ for which the \gls{ACN} system with $\ph T/2=\pi/K$ experiences the most severe loss in sum-\gls{SNR}. 
Thus, $\alpha^\star$ is a robust choice that allows us to avoid severe loss in sum-\gls{SNR} over a wide range of distances.

For short distances $|\dx|\leq 15$ $(\da=10\lambda)$, we see in \fig~\ref{fig:a_om:range} that $\aom T/2$ can take any value in $[-\pi/2,\pi/2]$ (recall that $\LA$ is periodic with $\pi$), which implies that for any $\ph\in \Astar$ inclduing $ \ph^\star$, there exists $\dv\in[-60,60]$ for which the most severe loss is experienced. In other words, no choice of phase slope is better than other in mitigating the effects of variation of $\om$ over the full speed range $\dv\in[-60,60]$ at short distances $|\dx|\leq 15$ $(\da=10\lambda)$. 
However, we recall that approximating $\om$ as an affine function~\eqref{eq:omega:app:affine} is not very accurate over the full range of speed difference at low distances. Therefore, a numerical quantification of the system performance at such short distances is shown in the numerical results section.

\section{Pathloss Variation Effects}\label{sec:PL}
In this section, we investigate the performance of \gls{ACN} under time-varying \gls{PL} following the same steps as in the previous section.
\subsection{Pathloss Model}
The Pathloss is typically modeled following a power law~\cite{AbbasPL2015}, and so does its inverse, the path gain, which will be the one explicitly used in the analysis to follow. A generic model of the path gain~\cite{AbbasPL2015} is given by
\begin{align}\label{eq:PL:generic:w:sahdowing}
	\Apl(t) = \Apl_0 \times \bigg(\frac{\Dplz}{\Dpl(t)}\bigg)^{n_{\mathrm{e}}} 10^{ -X_{\sigma}(\Dpl)/10}, ~\Dpl \geq \Dplz,
\end{align}
where $\Dpl(t) = \sqrt{ (\dx + \dv t)^2+\Lw^2 }$, is the distance between the \gls{Tx} and \gls{Rx}, $n_{\mathrm{e}}$ is the path loss exponent, $\Dplz$ is a reference distance with path gain $\Apl_0 $, and $X_{\sigma}(\Dpl)$ is a zero-mean Gaussian random process with standard deviation $\sigma_{\textrm{SH}}$ corresponding to shadowing (large scale fading). 
Shadowing is a spatially correlated process with autocorrelation $\mathbb{E} \{X_{\sigma} (\Dpl)X_{\sigma}(\Dpl + \Delta \Dpl)\}$,  that can be modeled using a decaying exponential function with parameter $d_{\mathrm{c}}$~\cite{Gudmundson1991}. The decorrelation distance $d_{\mathrm{c}}$, represents the distance difference at which the autocorrelation is equal to $\mathrm{e}^{-1}$. In context of highway scenario, it is reported to be  $ 23.3 \leq d_{\mathrm{c}} \leq 32.5~$m in~\cite{AbbasPL2015} and $d_{\mathrm{c}}=25~$m in WINNER+B1 pathloss model~\cite{3GPP_rel14_36p885}. Taking this into account, and for simplicity, we assume that shadowing is a block-type fading over the duration of $KT$s,
\begin{align}
X_{\sigma}  \triangleq X_{\sigma} (\Dpl |_{t=(K-1)T/2}) ,~ 0 \leq t \leq KT.\nonumber
\end{align}

Then, assuming that\footnote{When $\Dpl<\Dplz$ the model~\eqref{eq:PL:generic:w:sahdowing} is not valid. In such case, the path gain can be assumed to be equal to that at $\Dplz$ as done, for example, in WINNER+B1 model~\cite[\tab~A.1.4-1]{3GPP_rel14_36p885}.} $\Dpl(0)>\Dplz$ and $ \Dpl\big((K-1)T\big)>\Dplz$, we can approximate the average path gain following
\begin{align}
\overline{A}(t)&= \mathbb{E} \{ A(t) \} = \mu_{X}\Apl_0  \bigg( \frac{\Dplz}{\Dpl(t)}\bigg )^{n_{\mathrm{e}}}\approx \bpl +\apl t, \label{eq:Euler:PL}  
\end{align}
where $\mu_{X}= \mathbb{E} \{ 10^{-X_{\sigma}/10}\}  =\mathrm{e}^{ (\ln(10)\sigma_{\textrm{SH}})^2/200}$.
Assuming that the variation of path gain is negligible over a packet duration $\Tm \ll T$, we reach
\begin{align*}
  \overline{A}(t) \approx   \overline{A}(kT) = \bpl +\apl k T, ~ kT\leq t\leq kT+\Tm.
\end{align*}
The approximating affine function parameters can be computed using either first-order Taylor expansion or using the \gls{LS} method. Since $\overline{A}(kT)>0$, the coefficients need to satisfy the condition 
$\bpl+\apl kT>0$, $k=0, 1,\ldots, K-1$, for the approximation to be valid. When Taylor expansion at $t=t_0$ is applied, the coefficients are given by
\begin{align}
	\apl &= -\mu_{X}\Apl_0  \Dplz^{n_{\mathrm{e}}} \frac{n_{\mathrm{e}}\dv\times (\dx + \dv t_0)  }{\Dpl^{n_{\mathrm{e}}+2}(t_0)}, \label{eq:apl}
\end{align}
and $\bpl =\overline{\Apl} (t_0) -  \apl ~t_0$. 


\subsection{Sum-\gls{SNR}}
To study the effects of the variation of the average \gls{PL} on the \gls{ACN}/\gls{ABN} system, we assume that the two vehicles are moving in the same lane $\Lw=0$~m. This implies that the \gls{AOA} and the channel phase shift do not change over $KT$~s ($\aom =0$). 
Furthermore, we assume that the antennas have isotropic patterns (i.e., even if $\phi$ is time-varying, the antenna responses are fixed). 

The received power of the $k^{\text{th}}$ packet can be expressed as $P_{\textrm{r}}=\mathbb{E}\{|a(t)x(t)|^2\} =\overline{\Apl}(kT)P_{\mathrm{t}}$, where $\overline{\Apl}(kT)=\mathbb{E}\{|a(t)|^2\}$, and $P_{\mathrm{t}}=\mathbb{E}\{|x(t)|^2|\}$ is the transmitted power. Then, the \gls{SNR} of the $k^{\text{th}}$ packet~\eqref{eq:SNRk} is expressed as
\begin{align}
    \gamma_k =\frac{P_{\mathrm{t}}}{L\sigma^2_{\mathrm{n}}}  \overline{\Apl}(kT) \bigg|\sum_{l=0}^{L-1} \mathrm{e}^{-\jmath(\ps_l-\ph_l kT)} \bigg|^2,
\end{align}
since $g_l(\phi)=1$, $l=0,1$.
Then, normalizing the \gls{SNR} with respect to $P_{\mathrm{t}}/\sigma^2_{\mathrm{n}}$ and summing over $k$, we find that 
\begin{align}
\SsnrPL(x, y)
    &= \frac{1}{L} \sum_{k=0}^{K-1} \overline{A}(kT) \big(2 + 2\cos( y - 2 x k) \big)\nonumber \\
	&\approx \frac{1}{L} \sum_{k=0}^{K-1} (\bpl+\apl k T) \big(2 + 2\cos( y - 2 x k) \big)\nonumber\\
	& = \cJ K \bigg(1 +  \frac{\JPL(x, y)}{\cJ K} \bigg), \label{eq:S:S0+JPL}
\end{align}
where the approximation follows from~\eqref{eq:Euler:PL}, and where $x$ and $y$ are as defined in~\eqref{def:y:x} ($x=\ph T/2$, since $\aom =0$), $ \cJ  = \big(\bpl+\apl T (K-1)/2\big)$, and
\begin{align}
	\JPL(x, y)&=\sum_{k=0}^{K-1} (\bpl + \apl  k T )  \cos\big( y - 2x k\big).  \label{eq:JPL}
\end{align}

As done earlier, we define the loss function as
\begin{align}\label{eq:Loss:PL}
\LPL (x) \triangleq - \inf_{y\in[0,2\pi)} \frac{\JPL(x, y)}{\cJ K},
\end{align}
where the worst-case sum-\gls{SNR} is given by
\begin{align}\label{eq:S:OmLpl}
    \inf_{y\in [0,2\pi)} \SsnrPL(x, y) = \cJ K (1 - \LPL (x) ).
\end{align}
We then state certain properties of the loss function in the following lemma.

\begin{lemma}\label{lemma:PL} Let $\LPL (x) $ be as defined in~\eqref{eq:Loss:PL}, $x\in \mathbb{R}$, and $(\bpl+\apl k T)>0$, $k=0,1,\ldots,K-1$, then
\begin{enumerate}

\item The loss function is given by
\begin{align}
\LPL (x) = \begin{cases}
1, & x\in\mathcal{X}\\
\sqrt{ \cJ^2 f_1^2(x) + \cJJ ^2 f_2^2(x)} / (\cJ K), & x\notin\mathcal{X}
\end{cases}
\end{align}
where $ \cJ = \big(\bpl+\apl T (K-1)/2\big)$, $\cJJ =\apl T/2$, $\mathcal{X} $ is defined in~\eqref{def:Xstar}, $f_1$ is defined in~\eqref{eq:f1}, and  $f_2:\mathbb{R}\setminus{\mathcal{X}}\rightarrow \mathbb{R}$, is given by
\begin{align}
f_2(x)  &=   \frac{K\cos(Kx)}{\sin(x)}     -\frac{\sin(Kx)}{\sin(x)}\cot(x)  \label{eq:f2}. 
\end{align}

\item The loss function is periodic with period $\pi$ and symmetric around $\pi/2$.

\item $\LPL$ can be bounded as
	\begin{align}\label{eq:lem:PL:upper:lower}
	0 < \LPL (x) ,~ \apl\neq 0, \quad  \LPL (x)\leq 1.
	\end{align}
	\item  If $x\in \mathcal{X}^\star$ then
	\begin{align}\label{eq:lem:PL:xstar}
	\LPL (x) = \bigg | \frac{\cJJ }{\cJ \sin(x)} \bigg|.
	\end{align}
	
\item The loss function is bounded when $x\notin \mathcal{X}$ as
\begin{align}\label{eq:lem:PL:Bounds}
  \frac{|f_1(x)| }{K} \leq	\LPL (x)  \leq \frac{\sqrt{(K-1)^2f_1^2(x)+f_2^2(x)}}{K(K-1)}.
\end{align}

	\end{enumerate}
	
\end{lemma}
\begin{proof}
	See the Appendix.
\end{proof}
\begin{figure}
	\includegraphics[width= \columnwidth]{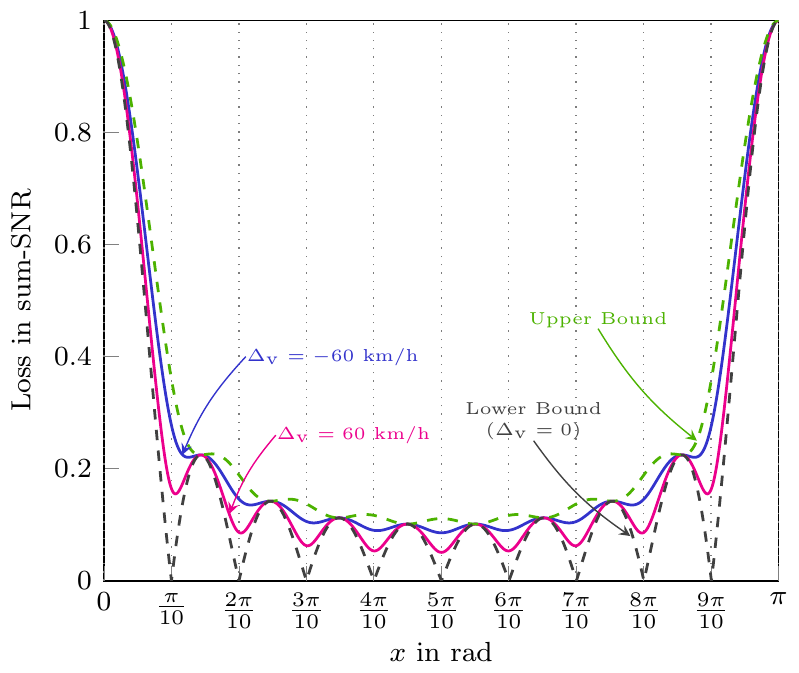}
	\caption{The sum-\gls{SNR} loss function under \gls{PL} variation, $\LPL(x)$, for $K=10$, $\dx= 30$~m, $\Lw = 0$~m, and $\dv=\pm 60 ~$ km/h. The upper and lower bounds~\eqref{eq:lem:PL:Bounds} of the loss function are shown in dashed lines.}
	\label{fig:lossFunction:PL}
\end{figure}

\subsection{Phase Slope Design Under Time-Varying \gls{PL} ($\apl\neq 0$)}
We recall that when the \gls{PL} is approximated as constant ($\apl = 0$), the phase slopes $ x = \ph T/2 \in \Xstar$ ensure identical optimal performance corresponding to $\inf_y \SsnrPL(x,y)= \cJ K$, and $\LPL (x) = 0$.
Now, when the \gls{PL} is time-varying, $\apl \neq 0$, we see from Lemma~\ref{lemma:PL}~\eqref{eq:lem:PL:xstar} that the loss in sum-\gls{SNR} $\LPL (x) \neq 0 $ for $x= \ph T/2 \in \Xstar$, and it is not identical for all the elements of $\Xstar$. 
Since $\LPL$ is periodic with $\pi$, there exists $K-1$ unique elements in $\Xstar$, which are represented in $\Astar$~\eqref{def:Astar}. 
Following that, we can deduce that the phase slope that minimizes the loss in sum-\gls{SNR} under \gls{PL} variation among $\Astar$ is also the solution to~\eqref{eq:design:rule}. That is,
\begin{align}\label{eq:rulePL:ruleOmega}
\arg \min_{\ph\in \Astar} \LPL (\ph T/2) = \arg \min_{\ph \in \Astar} |\ph T/2 -\pi/2|.
\end{align}
This follows from~\eqref{eq:lem:PL:xstar} since $1/|\sin (\ph T/2)|$ is minimized by the phase slope satisfying the right-hand side of~\eqref{eq:rulePL:ruleOmega}.
These results can be observed in \fig~\ref{fig:lossFunction:PL} where the loss function,
is plotted for $\dx=30~$m and $\dv=\pm 60~ $km/h. 
The WINNER+B1~\cite{3GPP_rel14_36p885,WINNER+} pathloss model is used. Assuming a carrier frequency $f_{\textrm{c}}=5.9~$GHz, and antenna heights of $1.5~$m, the path gain between the \gls{Tx} and \gls{Rx} when $\Dpl\leq 177~$m (this distance depends on $f_{\textrm{c}}$ and the antennas heights) can be modeled using~\eqref{eq:PL:generic:w:sahdowing}, with the parameters, $\Dplz = 3~$m, $\Apl_0= 10^{5.32}$, $n_{\mathrm{e}} = 2.27$, and $\sigma_{\textrm{SH}}=3$.

From \fig~\ref{fig:lossFunction:PL}, and using the bounds in~\eqref{eq:lem:PL:Bounds}, we see that the loss in sum-\gls{SNR} is low for all $\ph \in \Astar$ ($x = \ph T/2 = q \pi/K$). In particular, substituting by the upper bound~\eqref{eq:lem:PL:Bounds} of the loss function in~\eqref{eq:S:OmLpl} we can deduce that the phase slopes $\ph = q 2\pi/KT$, where $q =1,\ldots, K/2$, ($\ph \in \mathcal{A}^\star$), respectively achieve a sum-\gls{SNR} that is at worst, $-1.94$, $-0.91$, $-0.64$, $-0.54$, and $-0.51$~dB, lower than the reference zero-loss sum-\gls{SNR}, $\inf_y \SsnrPL (x,y) = \cJ K$. Note that due to the symmetry of the loss function with respect to $\pi/2$, the phase slopes $\ph= q 2\pi/KT$, $q = K/2+1,\ldots, K-1$, achieve identical sum-\gls{SNR} to their symmetric counterparts.

Thus unlike the effects of variation in $\om$, when \gls{PL} is time-varying over the duration of $KT~$s at medium to large distances, the loss in sum-\gls{SNR} is much lower than the maximum loss $\LPL (x) = 1$, and thus \gls{PL} time variation cannot cause severe loss in sum-\gls{SNR}.

As $\dx$ increases, the loss in sum-\gls{SNR} decreases for all phase slopes $\ph\in \Astar$. This can be observed from Lemma~\ref{lemma:PL}~\eqref{eq:lem:PL:xstar}, where the loss is proportional to $\cJJ/\cJ$, which is given by
\begin{align}
    \bigg|\frac{\cJJ}{\cJ}\bigg| = \frac{n_{\mathrm{e}} |\dv|}{|\dx + \dv t_0|}T/2, \quad (\Lw = 0),
\end{align}
assuming the use of~\eqref{eq:apl} with $t_0 = (K-1)T/2$. As $\dx$ becomes large, the loss becomes negligible for all phase slopes $\ph\in \Astar$.

\subsection{Combined \gls{PL} and $\om$ Variation}\label{sec:combined:PL+phi}
Consider now the scenario when the two \glspl{VU} are on different lanes $\Lw\neq 0$. Then, besides \gls{PL}, $\om$ is also time-varying over $KT$~s, and can be approximated using~\eqref{eq:omega:app:affine}, where $\aom\neq 0$. Despite that $\phi$ varies too, the isotropic antenna responses are fixed, and do not affect the system. In this scenario, the sum-\gls{SNR} and the loss function are still modeled by~\eqref{eq:S:S0+JPL} and~\eqref{eq:Loss:PL} respectively, where $x=(\ph - \aom)T/2$. 
Therefore, at medium to long distances, variation in \gls{PL} changes the shape of the loss function, in the sense that, for any $x$, the loss in sum-\gls{SNR} is greater than the loss under non-varying \gls{PL}, $\LPL(x)\geq \LPL(x)|_{\dv=0}$, as can be seen in \fig~\ref{fig:lossFunction:PL}. On other hand, variation in $\om$ introduces deviation to the \gls{ACN}/\gls{ABN} phase slope $\ph$ with a factor of $\aom$.
Using the bounds~\eqref{eq:lem:PL:Bounds} and from \fig~\ref{fig:lossFunction:PL}, we see that the loss function under \gls{PL} variation has similar trends to the loss function under time-varying $\om$~\eqref{eq:lemma:loss:cases} (which is the lower bound). That is, the maximum loss occurs at the points $x=(\ph - \aom)T/2=q\pi$, and the loss function has decreasing maxima within $[\pi/K,\pi/2]$, and increasing maxima in $[\pi/2,(K-1)\pi/K]$.
Thus, under both \gls{PL} and $\om$ variation, the design rule~\eqref{eq:design:rule} still leads to the most robust phase slope in $\Astar$. In other words, the phase slope~\eqref{eq:design:rule} allows us to avoid the most severe loss in sum-\gls{SNR} over the widest range of $\aom$ when both $\om$ and \gls{PL} are time-varying over $KT$~s.

\section{Effects of \gls{AOA} Variation With Non-Isotropic Antennas}\label{sec:phi}

So far, we assumed that the antenna patterns are isotropic. Now, consider non-isotropic antennas with far-field functions $g_0(\phi)$, $g_1(\phi)$.
When the two reference vehicles are on different lanes ($\Lw \neq 0$), the \gls{AOA} varies depending on the speed and initial distance between the \gls{Tx} and \gls{Rx}. That introduces a variation in both the gain and the phase of the far-field functions of antennas.
In the following, we investigate the effects of these variations on the performance of \gls{ACN} at medium to large distances.
As done in the previous sections, we study the effects of variation of antenna far-field functions while neglecting the effects of the two other quantities of interest ($\Omega$ and \gls{PL}), by assuming that they are both non-varying over $KT$~s. That is, $\Omega(t)\approx \bom$, $\aom =0$ in~\eqref{eq:omega:app:affine}, and $\overline{A}(t)\approx \bpl $, $\apl=0$ in~\eqref{eq:Euler:PL}, when $0 \leq t\leq KT$.
Then, we express the \gls{SNR} of the $k^{\text{th}}$ packet~\eqref{eq:SNRk} in this case as
\begin{align}
    \gamma_k =\frac{ P_{\mathrm{t}}\bpl}{L\sigma^2_{\mathrm{n}}}  \bigg|\sum_{l=0}^{L-1} |g_l(\phi_k)|\mathrm{e}^{-\jmath(\ps_l(kT)-\ph_l kT)} \bigg|^2,
\end{align}
where the \gls{AOA} is assumed to be constant over a packet duration $\phi(t) = \phi_k$,  $kT\leq t\leq kT+\Tm$, $k=0, 1,\ldots,K-1$, and so is $\ps_l$, $l=0,1$.
Recalling that $\ph_0=0$, and $\ph = \ph_1 \in \mathbb{R}$, the normalized sum-\gls{SNR} can be expressed as
\begin{align}
\frac{\sigma^2_{\mathrm{n}}}{ P_{\mathrm{t}}\bpl } \sum_{k=0}^{K-1} \gamma_k&=
 \sum_{k=0}^{K-1} \sum_{l=0}^{1} \frac{|g_l(\phi_k)|^2}{2}  + \sum_{k=0}^{K-1} |g_0(\phi_k)g_1(\phi_k)| \nonumber\\
 \times\cos\big( &\psi_1(kT)-\psi_0(kT) - \ph kT  \big).\label{eq:Sphi}
\end{align}
From~\eqref{def:psi} it follows that at $ kT \leq t\leq kT+\Tm$
\begin{align*}
    \psi_1(kT) - \psi_0(kT) &= \Omega(kT) -\phase{g_1 (\phi_k)} + \phase{g_0(\phi_k)} - \beta \nonumber\\
&= \bom  -\phase{g_1 (\phi_k)} + \phase{g_0(\phi_k)}  - \beta,
\end{align*}
since $\Omega_0(t) = 0$, $\Omega_1(t)=\Omega(t)=\bom$ when $0\leq t\leq KT$, $ \beta_0=0$ and $\beta_1 = \beta \in [0,2\pi)$.
Unlike the case when antennas were assumed to be isotropic, we see here that the phase responses of non-isotropic antennas change with $\phi_k$, and introduce a time-varying phase shift. To analyze its effect we apply an affine approximation taking into account the slow change of the \gls{AOA} at medium to large distances (e.g., at $\dx=30~$m, $\phi$ changes by $9\deg/1~$s at most, when $\dv\in[-60,60]~$km/h), following
\begin{align}\label{eq:approx:DeltaPhase}
  \phase{g_0(\phi_k)} -\phase{g_1 (\phi_k)}  \approx \bgph + \agph kT.
\end{align}
The coefficients $\bgph$ and $\agph$ depend on the antenna patterns used and how fast the \gls{AOA} changes. They can be obtained using the \gls{LS} method or using the first-order Taylor series when an analytical formula for the antenna far-field functions is available.
Adopting the above approximation, we can express the sum-\gls{SNR}~\eqref{eq:Sphi} as
\begin{align}
    \SsnrPH(x,y) 
    &=\sum_{k=0}^{K-1} \sum_{l=0}^{1} \frac{|g_l(\phi_k)|^2}{2}  + \JPH(x,y), \nonumber
\end{align}
where
\begin{align}
   \JPH(x,y)= \sum_{k=0}^{K-1} &|g_0(\phi_k)g_1(\phi_k)| \cos( y - 2 x k),
\end{align}
 $x = (\ph-\agph )T/2$ and $y= (\bom +\bgph -\beta)$.
 The \gls{ACN} initial phase shift $\beta$ can take any value in $[0,2\pi)$, and so does $y $. Therefore, we assess the performance according to worst-case sum-\gls{SNR}
 \begin{align}\label{eq:Sphi:inf}
     \inf_{y\in[0,2\pi)} \SsnrPH (x,y) = \sum_{k=0}^{K-1} \sum_{l=0}^{1} \frac{|g_l(\phi_k)|^2}{2}  + \inf_{y\in[0,2\pi)}  \JPH(x,y).
 \end{align}
 
 
 When developing the \gls{ACN} system~\cite{ACN,lehocine2021abn} under worst-case propagation assumptions, the  signal direction $\phi_k$ was assumed to be negligibly varying over $KT$~s, and coinciding with worst-case \gls{AOA} defined as 
\begin{align}\label{eq:phi:min}
\phim= \arg \inf_{\phi\in[0,2\pi)} ~ \sum_{l=0}^{1} \frac{|g_l(\phi)|^2}{2}.
\end{align}
 Thus, we can straightforwardly establish that when the \gls{AOA} is varying, we achieve a gain in sum-\gls{SNR} with respect to the worst-case scenario (i.e., $\phi_k=\phim$, $\forall k$) since the first term in~\eqref{eq:Sphi:inf} satisfies
\begin{align}\label{eq:phi:var:gain}
		\sum_{k=0}^{K-1} \sum_{l=0}^{1} \frac{|g_l(\phi_k)|^2}{2} &\geq
	 \sum_{k=0}^{K-1} \sum_{l=0}^{1} \frac{|g_l(\phim)|^2}{2}\nonumber\\&= K \sum_{l=0}^{1} \frac{|g_l(\phim)|^2}{2}.
\end{align}
On the other hand, under worst-case assumptions of non-varying $\phi_k$ for $k=0,\ldots, K-1$ (which implies that $\agph=0$), we know that for $x=\ph T/2\in \Xstar$, $\inf_y \JPH(x,y) = 0$. To investigate if this holds when the \gls{AOA} is time-varying, and taking into account the slow change rate of $\phi$ at medium to large distances, we approximate using the \gls{LS} or first-order Taylor series
\begin{align}\label{eq:phi:lin:model}
	|g_0(\phi_k)g_1(\phi_k)| \approx \aphiz+ \aphi kT ,
\end{align}
where $\aphiz+ \aphi kT >0$, $k=0,1,\ldots,K-1$.
Following that,
\begin{align}\label{eq:J:phi:app}
\JPH(x,y) & =\sum_{k=0}^{K-1} |g_0(\phi_k)g_1(\phi_k)| \cos( y - 2 x k)\nonumber\\ &\approx\sum_{k=0}^{K-1} (\aphiz + \aphi kT )\cos( y - 2 x k). 
\end{align}
From~\eqref{eq:J:phi:app} we see that variation in $\phi$ has identical effects to \gls{PL} variation effects captured by $\JPL$~\eqref{eq:JPL}, and therefore they can be represented using the same loss function. Consequently, using the left-hand statement of Lemma~\ref{lemma:PL}~\eqref{eq:lem:PL:upper:lower}, and the definition~\eqref{eq:Loss:PL}, we deduce that
\begin{align}
    \inf_{y \in [0,2\pi)}\JPH(x,y)<0, \quad \aphi\neq 0, x\in \mathbb{R} \nonumber.
\end{align}
Hence, when the \gls{AOA} is time-varying we also experience a loss in sum-\gls{SNR}~\eqref{eq:Sphi:inf} (unlike the case when $\phi$ is fixed and where $\inf_y \JPH(x,y)=0$). The loss in sum-\gls{SNR} is due to two factors. The first factor is the variation of the antenna gains $|g_0(\phi_k)g_1(\phi_k)|$, which results in a loss even if $\agph=0$, $x = \ph T/2  \in \Xstar$. This is identical to the effects of  
\gls{PL} variation. The second factor is due to the variation of the phase response of the antennas~\eqref{eq:approx:DeltaPhase}, which shifts the effective phase slope of \gls{ACN} to $x = (\ph -\agph) T/2$, and introduces a loss in sum-\gls{SNR} that that has been studied in \secR~\ref{sec:omega} under time variation of $\Omega$.

In summary, the effects of time-varying \gls{AOA}, with affine approximation of phase and magnitude responses of non-isotropic antennas, have the same model as the combined \gls{PL} and $\Omega$ time variation model that has been discussed in~\secR~\ref{sec:combined:PL+phi}. Therefore, we can readily conclude that the phase slope in $\Astar$ that ensure a robust performance under the effects of time-varying far-field functions is given by~\eqref{eq:design:rule}.
We recall that these conclusions apply at medium to large distances, where the affine approximations used in evaluating the sum-\gls{SNR} loss term, $\inf_y \JPH(x,y) $, are expected to be accurate.
At short distances, despite that the characterizations of the loss term made in this section may not be very accurate, the gain term in sum-\gls{SNR} indicated by~\eqref{eq:phi:var:gain} is still achieved. Furthermore, this gain is expected to be more noticeable at short rather than large distances, since the \gls{AOA} change rate is much higher for the former case. These aspects are highlighted using an example of an antenna pattern in the numerical results section.

\color{black}
\begin{figure}
	\includegraphics[width= \columnwidth]{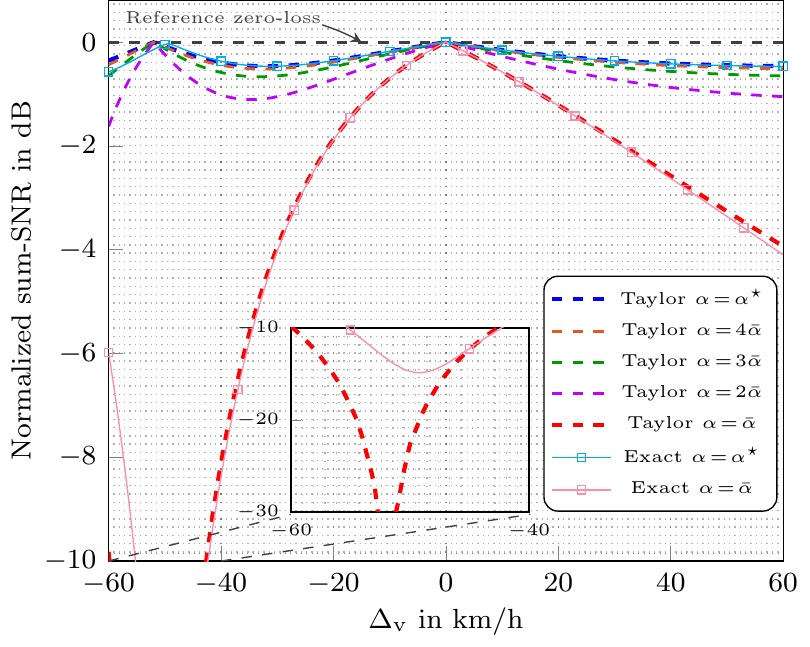}
	\caption{Normalized worst-case sum-\gls{SNR} under $\om$ variation $\big(1-\LA(x)\big)$ in dB, for $K=10$, $\da=10\lambda$, $\dx=30~$m. ($\phBar=\frac{2\pi}{KT}$)}
	\label{fig:sum-SNR:K10}
\end{figure}

\section{Numerical Results}
In this section, we present numerical results corresponding to the normalized sum-\gls{SNR} of a $1\times 2$ \gls{ACN} system (or equivalently $2\times 1$ \gls{ABN}) when $\om$, \gls{PL}, and $\phi$ vary over the duration of a burst of $K$ consecutive \glspl{CAM}, $KT$~s.
Throughout this section, we assume that the \gls{CAM} broadcast period $T=0.1~$s, and that the maximum tolerable \gls{AoI} is proportional to $K=10$ packets.

\subsection{Sum-\gls{SNR} Under Variation of $\om$}
In the following, we visualize the sum-\gls{SNR} for different choices of $\ph\in \Astar$, taking into account the time variation of $\om$ only. That is, the antennas are assumed to be isotropic, and the change in \gls{PL} is assumed negligible over $KT$~s. The coefficients for the affine approximation of $\om$ in~\eqref{eq:omega:app:affine} are computed based on Taylor series expansion~\eqref{eq:a:omega:taylor}. The \gls{LS}-based affine approximation was found to lead to comparable results, and it is therefore omitted in the figures to follow in this subsection.

We recall that the worst-case normalized sum-\gls{SNR}~\eqref{eq:S:min}, is given by
\begin{align}\label{eq:NR:Som}
  \inf_{y\in [0,2\pi)}\SsnrA(x,y)= K\big(1- \LA(x)\big),  
\end{align}
where $x=(\ph-\aom)T/2$.
Since $\aom$~\eqref{eq:a:omega:taylor} depends on $\dv,\dx,\da$, and $\Lw$, we visualize the sum-\gls{SNR} as a function of $\dv$ for a fixed initial distance $\dx$, and a fixed antenna separation $\da$. Furthermore, since the effect of lateral distance results in $\aom|_{-\Lw}=-\aom|_{\Lw}$ (follows from~\eqref{eq:omega:full}, or~\eqref{eq:a:omega:taylor}), the sum-\gls{SNR} is minimized over $\Lw \in \{-4,4\}$.

\begin{figure}
	\includegraphics[width= \columnwidth]{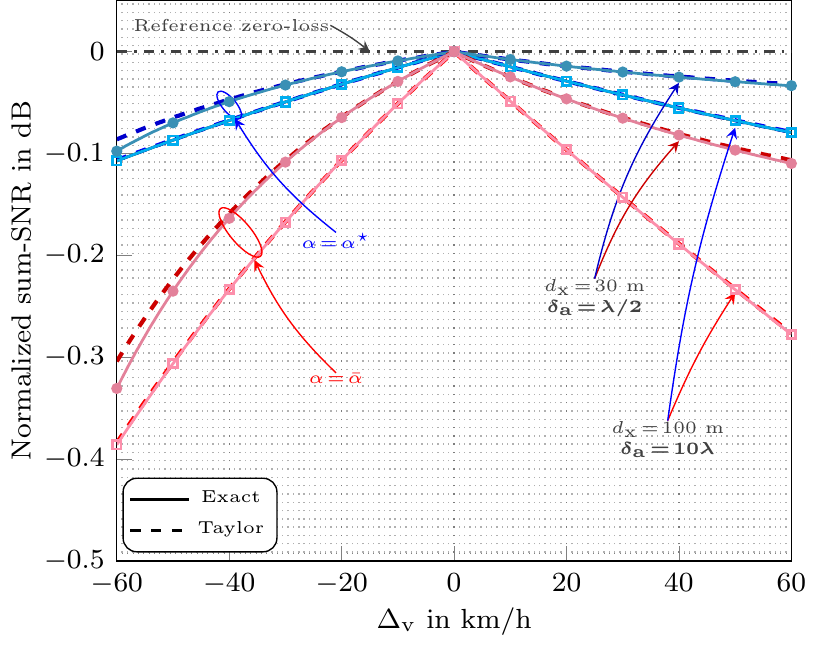}
	\caption{Normalized worst-case sum-\gls{SNR} under $\om$ variation $\big(1-\LA(x)\big)$ in dB, for $K=10$, $\da=10\lambda$ at $\dx=100~$m, and  $\da=\lambda/2$ at $\dx=30~$m. ($\phBar=\frac{2\pi}{KT}$)}
	\label{fig:sum-SNR:d100:1lam:d30}
\end{figure}


In \fig~\ref{fig:sum-SNR:K10}, we plot the sum-\gls{SNR}, after further normalization with respect to $K$, in dB, at a medium distance $\dx = 30$~m, when the antenna separation is $\da=10\lambda$. 
We recall that the loss function under variation of $\om$, $\LA$ is symmetric around $\pi/2$ and therefore, we visualize the results only for $\ph\in \Astar$, $\ph = q \frac{2\pi}{KT}$, $q= 1,\ldots, 5$. From \fig~\ref{fig:sum-SNR:K10}, we see that the choice of $\ph$ satisfying~\eqref{eq:design:rule}, i.e., $\alpha^\star= 5 \times 2\pi/KT$, exhibits a robust performance to channel phase variation
with a low loss over the full range of $\dv\in[-60,60]~$km/h. 
 Besides $\alpha^\star$, the choice of $\ph=4\times 2\pi/KT$, or $\ph= 3\times 2\pi/KT$ results in comparably low loss too, which follows since the phase slopes have, respectively, the second and the third largest phase distances from the points of severe sum-\gls{SNR} loss $x=q\pi$. The phase slope having the shortest phase distance from these points, $\phBar= 2\pi/KT$, exhibits the highest loss, and it is the least robust phase slope in $\Astar$.
In \fig~\ref{fig:sum-SNR:K10}, we also plot the sum-\gls{SNR} using the exact value of $\om$~\eqref{eq:omega:full}, for the most and the least robust choices of phase slopes in $\Astar$  ($\alpha^\star$ and $\phBar$). We see that the sum-\gls{SNR} using the linearized model is matching the sum-\gls{SNR} using the exact value of $\om$ for most speeds. 
The maximum loss experienced by $\alpha^\star$ when antennas are $\da=10 \lambda$ apart, is approximately $0.5~$dB (according to the exact sum-\gls{SNR}).
If the antenna separation is smaller $\da=\lambda/2$, then the \gls{ACN} experiences a significantly lower loss in sum-\gls{SNR} for both $\alpha^\star$ and $\phBar$ as can be seen in \fig~\ref{fig:sum-SNR:d100:1lam:d30}. In fact, the loss in sum-\gls{SNR} of the \gls{ACN} with $\da=\lambda/2$ at $\dx=30$~m is comparable to that of the \gls{ACN} with $\da=10\lambda$ at a large distance $\dx = 100~$m. This implies that implementing \gls{ACN} systems with low antenna separation limits the sum-\gls{SNR} loss incurred due to time variation of $\om$. 
Note that for both the cases shown in \fig~\ref{fig:sum-SNR:d100:1lam:d30}, $\alpha^\star$ has an advantage compared to $\phBar$, yet the loss in sum-\gls{SNR} is not significant for both phase slopes.

\begin{figure}
	\includegraphics[width=\columnwidth]{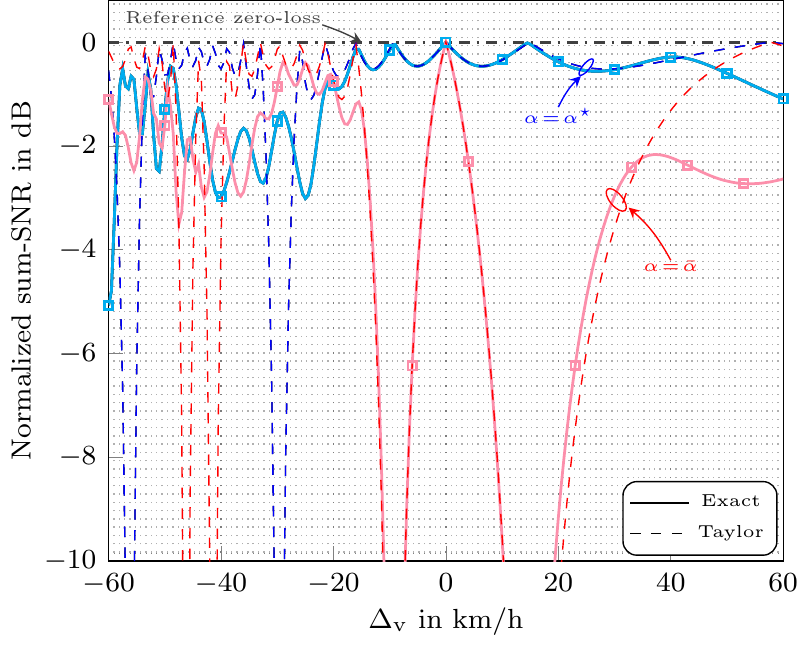}
	\caption{Normalized worst-case sum-\gls{SNR} under $\om$ variation $\big(1-\LA(x)\big)$ in dB, for $K=10$, $\da=10\lambda$, $\dx=10~$m. ($\phBar=\frac{2\pi}{KT}$)}
	\label{fig:sum-SNR:d10}
\end{figure}

In \secR~\ref{sec:sub:omega:var} it was explained that modeling $\om$ as an affine function holds under medium to large distances. Therefore, we are interested in evaluating sum-\gls{SNR} under short distances using the exact value of $\om$. In \fig~\ref{fig:sum-SNR:d10} we show the sum-\gls{SNR} at $\dx=10~$m, when the antenna separation is $\da=10\lambda$. As expected we see that at such a short distance the linearized model is accurate only for relatively low speeds $\dv \in [-15,35]$km/h. We note that for negative relative speeds the variation in $\om$ is faster than for positive speeds since, for the former, the geometries between the \gls{Tx} and \gls{Rx} change more dramatically due to decreasing distance between the \glspl{VU} and potential overtaking movement. For the latter, the distance between the \glspl{VU} is increasing and thus the change in geometries is less severe. 
At larger speed differences, the change in $\om$ is much faster and it cannot be approximated using a first-order polynomial. Following this at the points $x=(\ph -\aom)T/2=q\pi$, as it is the case for the \gls{ACN} with $\alpha^\star$ at $\dv\approx -30,-58~$km/h, the loss in sum-\gls{SNR} is not as severe as it is predicted using the linearized model. This is also observed at $\dv=-41,-46$km/h for $\phBar$.  
Thus, when the change in $\om$ is very fast and cannot be accurately approximated using a first-order polynomial, the \gls{ACN}/\gls{ABN} system does not exhibit a severe loss in sum-\gls{SNR} as predicted using the linearized model. The system exhibits better performance. 
We note that at this short distance $\alpha^\star$ shows advantage compared to $\phBar$ mostly when $\dv>-15$~km/h.  

To highlight more the advantage of a proper choice of $\ph$, in \fig~\ref{fig:sum-SNR:d10:lambda} we show the sum-\gls{SNR} at $\dx=10~$m, when $\da=\lambda/2$. We can observe that for such small antenna separation the linearized model is valid over a larger range of speed even at such a short distance. Moreover, the phase slope $\alpha^\star$ sustains a performance with low loss over the full speed range compared to $\phBar$.

\begin{figure}
	\includegraphics[width=\columnwidth]{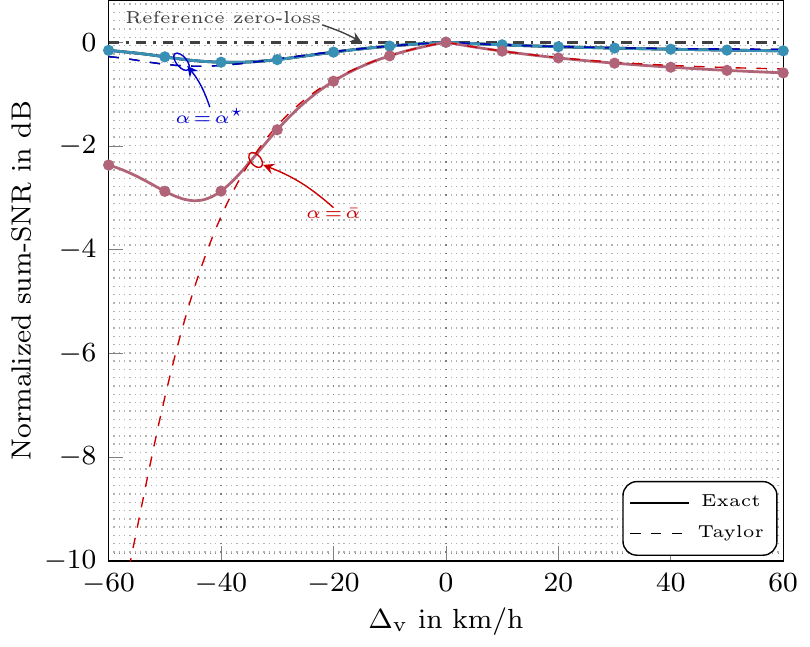}
	\caption{Normalized worst-case sum-\gls{SNR} under $\om$ variation $\big(1-\LA(x)\big)$ in dB, for $K=10$, $\da=\lambda/2$, $\dx=10~$m. ($\phBar=\frac{2\pi}{KT}$) }
	\label{fig:sum-SNR:d10:lambda}
\end{figure}

\subsection{Sum-\gls{SNR} Under \gls{PL} Variation}
To show the effects of \gls{PL} variation and the performance of the different phase slopes, we plot in \fig~\ref{fig:sum-SNR:PL} the sum-\gls{SNR}~\eqref{eq:S:OmLpl} for the worst-case $y$ as a function of $\dv$ at a fixed medium distance $\dx=30~$m. Both vehicles are on the same lane, $\Lw=0~$m, which results in non-varying $\om$, ($\aom =0$). Antennas are assumed to be isotropic. 
We normalize the sum-\gls{SNR} with respect to the average path gain at $\dx=30~$m, 
and with respect to $K$, such that it is equal to $0$~dB at $\dv=0$~km/h. 
We show the results for the affine approximation of the path gain, which is modeled according to WINNER+B1 model, based on both Taylor series~\eqref{eq:apl} and the \gls{LS} method, since the latter exhibited more accurate results in this case. 
From the figure, we observe that the loss due to \gls{PL} variation is not significant for all phase slopes $\ph\in \Astar$. In particular, the loss is at most around $0.4$~dB, and $1.5$~dB, for $\ph^\star T/2= \pi/2$, and $\phBar T/2=\pi/K$, respectively (based on exact curves at $\dv=-60$~km/h). The loss of the remaining phase slopes within $\Astar$ lies in between.
As discussed in \secR~\ref{sec:PL}, unlike the effects of variation of channel phase, variation in \gls{PL} does not introduce a shift to the \gls{ACN}/\gls{ABN} phase slope, and this explains the non-significant loss in sum-\gls{SNR}. 
\begin{figure}
		\includegraphics[width=\columnwidth]{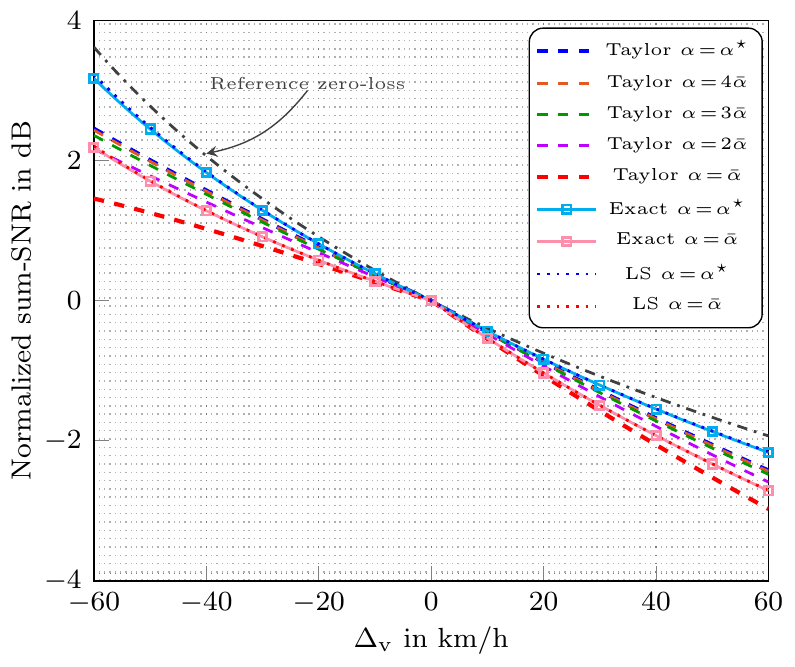}
	\caption{Normalized worst-case sum-\gls{SNR} under \gls{PL} variation in dB, for $K=10$, $\dx=30~$m, $\Lw=0~$m. The \gls{PL} model parameters correspond to WINNER+B1 model, and they are $\Dplz= 3~$m, $\Apl_0= 10^{5.32}$, $n_{\mathrm{e}}=2.27$, and $\sigma_{\textrm{SH}}=3$. ($\phBar=\frac{2\pi}{KT}$)}
	\label{fig:sum-SNR:PL}
\end{figure}

\begin{figure*}
\centering
\subfigure[Magnitude response $|g_l(\phi)|^2$ in dBi]{\includegraphics[width=0.8\columnwidth]{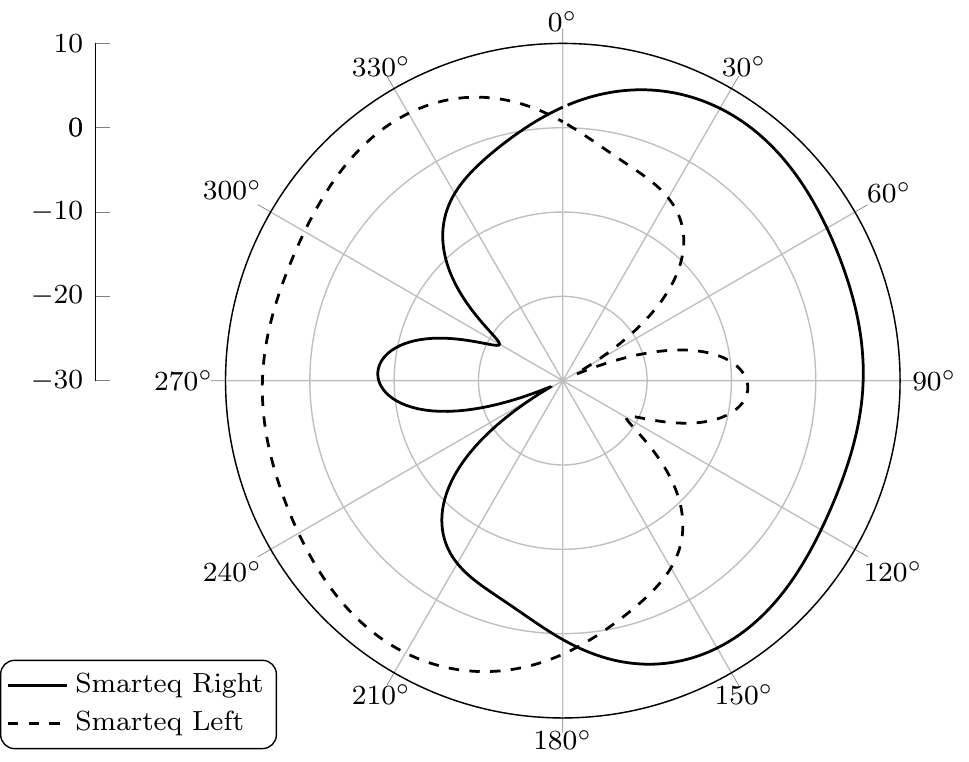}}
	\subfigure[Phase response $\phase{g_l(\phi)}$ in degrees]{\includegraphics[width=0.8\columnwidth]{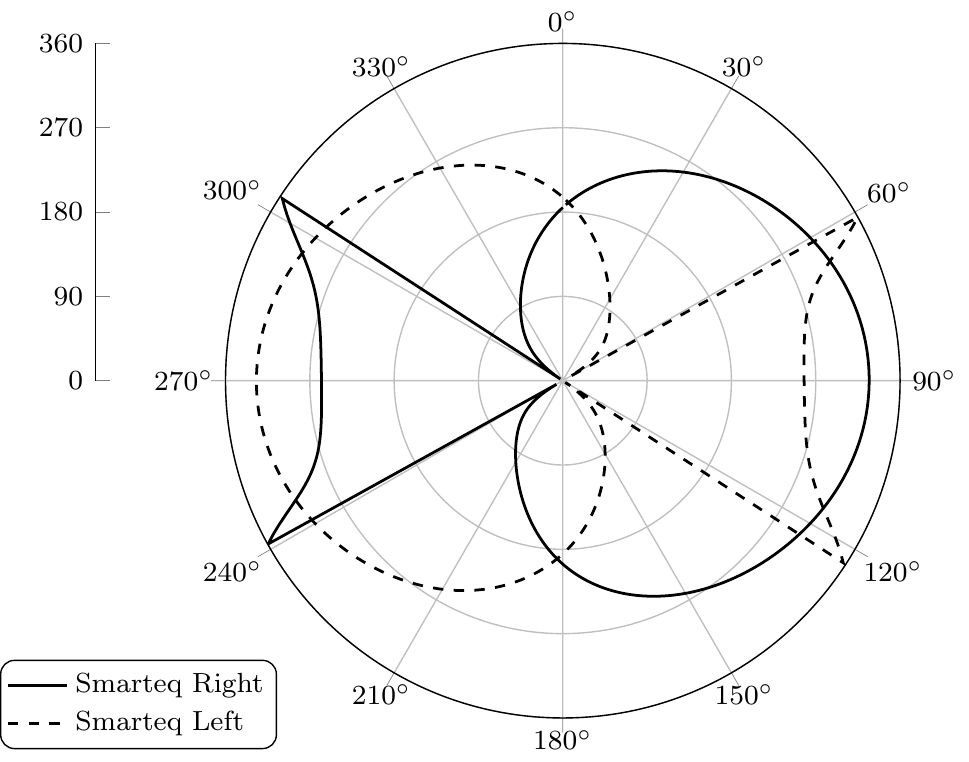}}
	\caption{Back-to-back patch antenna pattern. The antennas are designed by Smarteq\protect\footnotemark for Vehicular applications.}
	\label{fig:antenna}
\end{figure*}

\subsection{Sum-\gls{SNR} Under Antenna Response Variation}
In \secR~\ref{sec:phi}, we have seen that as $\phi$ varies, the phase and amplitude responses of antennas vary too. 
In \fig~\ref{fig:sum-SNR:Phi:smarteq} we show the sum-\gls{SNR}~\eqref{eq:Sphi:inf} under these effects when the antennas employed by the \gls{ACN}/\gls{ABN} are the ones shown in \fig~\ref{fig:antenna}.  
 The sum-\gls{SNR} is plotted as a function of $\dv$ for a fixed $\dx$, and for the worst-case $y\in[0,2\pi)$, and worst-case $\Lw \in \{-4,4\}$. 
The received signal is assumed to coincide with the worst-case \gls{AOA}, $\phim$~\eqref{eq:phi:min} at time $t=0$. Moreover, \gls{PL} and $\Omega$ are assumed to be approximately constant over $KT$~s. The affine coefficients for the approximation of far-field functions phase~\eqref{eq:approx:DeltaPhase} and amplitude~\eqref{eq:phi:lin:model} are computed using the \gls{LS}. For consistency, we normalize the sum-\gls{SNR} with respect to $( |g_0(\phim)|^2 + |g_1(\phim)|^2)$, and with respect to $K$, such that the sum-\gls{SNR} equals $0$~dB at $\dv=0$~km/h.

From \fig~\ref{fig:sum-SNR:Phi:smarteq} we see that at a medium distance $\dx=30~$m, the loss in sum-\gls{SNR} is low for $\phBar$, and negligible for $\alpha^\star$ and the remaining phase slopes in $ \Astar$. Despite that the phase slopes are shifted by $\agph$, resulting in $x=(\ph-\agph)T/2$, the loss in sum-\gls{SNR} is not as large as it was the case under time variation of $\om$ at $\dx=30~$m and $\da=10\lambda$. 
The sum-\gls{SNR} gain~\eqref{eq:phi:var:gain} due to variation of $\phi$ is negligible at $\dx=30~$m for these antenna patterns, as can be inferred from the reference curve in \fig~\ref{fig:sum-SNR:Phi:smarteq} ($\dx=30$~m). However, at $\dx=10~$m, this gain is higher.  
On the other hand, the loss in sum-\gls{SNR} due to $\inf_y \JPH(x,y)$ is also higher at $\dx=10$~m compared to $\dx=30$~m, and it is most noticeable for $\phBar$.
This loss is mostly attributed to the phase deviation $\agph$, which causes dips in sum-\gls{SNR} e.g., at $\dv \approx 38$~km/h, when $x=(\phBar-\agph)T/2=q\pi$. This can be concluded by looking at the sum-\gls{SNR} for $\phBar$, when the phase response of the antennas is assumed constant over $KT$~s, ($\agph=0$).
Unlike $\phBar$, the phase slopes $\alpha^\star$, and $\ph\in\{4\times 2\pi/KT,3\times 2\pi/KT\} $, effectively mitigate these effects and exhibit a low loss (approximately less than $0.5$~dB) over the full range of $\dv$.

Despite that the linearized model cannot be generally claimed to be accurate at short distances, we see from \fig~\ref{fig:sum-SNR:Phi:smarteq} that it does predict the actual behavior of the sum-\gls{SNR} curves, and it allows us to correctly select a robust phase slope within $\Astar$~\eqref{eq:design:rule} in mitigating the loss in sum-\gls{SNR} due to the variation of antenna responses.

 
\footnotetext{Smarteq Wireless AB, Sweden is a company specializing in
	antenna design and development for vehicle industry and others.}
\begin{figure}
	\includegraphics[width=\columnwidth]{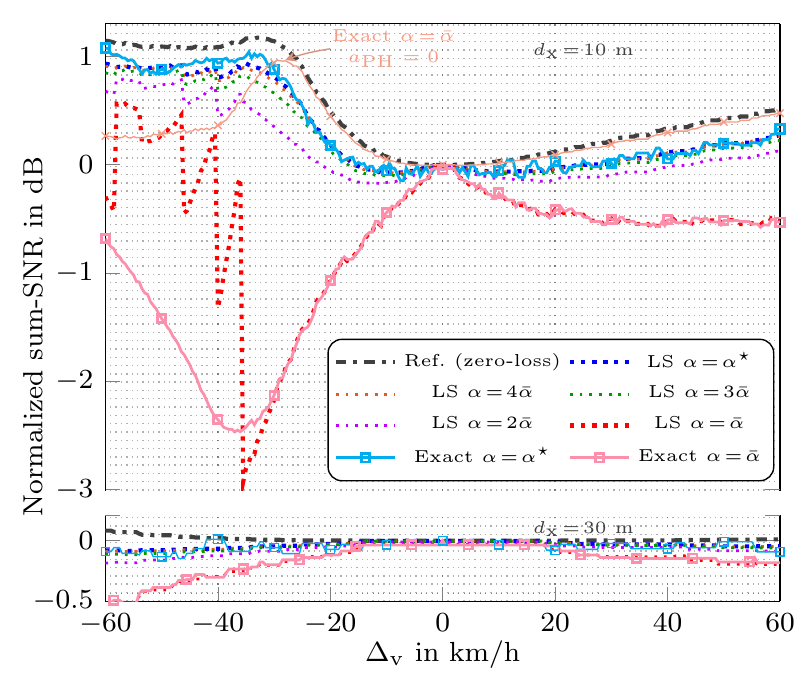}
	\caption{Normalized worst-case sum-\gls{SNR} under antenna response variation in dB, for $K=10$, $\dx = 30~$m, and $\dx = 10~$m. 
	($\phBar=\frac{2\pi}{KT}$) }
	\label{fig:sum-SNR:Phi:smarteq}
\end{figure}
\begin{figure}
	\includegraphics[width=\columnwidth]{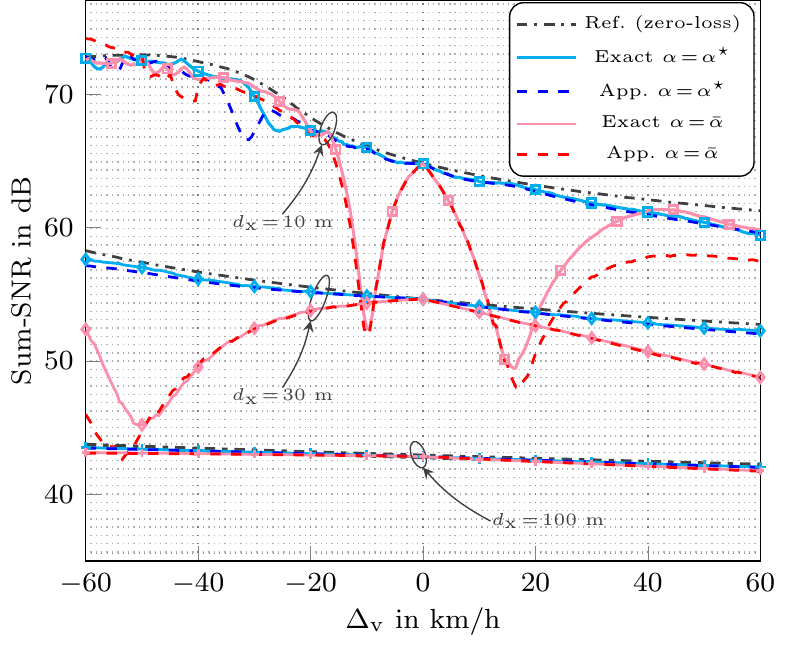}
	\caption{Worst-case (unnormalized) sum-\gls{SNR} under variation of $\Omega$, \gls{PL}, and \gls{AOA} in dB, for $K=10$, $\da=10\lambda$,  $P_{\mathrm{t}}=23$~dBm, and $\sigma_{n}^2= -95$~dBm. The linear approximation coefficients of $\Omega$ and \gls{PL} are computed using Taylor expansion, and those of antenna responses are computed using the \gls{LS}.  ($\phBar=\frac{2\pi}{KT}$) }
	\label{fig:sum-SNR:allthree}
\end{figure}
\subsection{Combined Effects of the Dominant Path Variation}
In this section, we show the sum-\gls{SNR} taking into account the combined effects of time-varying $\Omega$, \gls{PL}, and $\phi$. The \gls{PL} is modeled using WINNER+B1 channel model, and the receiver antennas have the patterns shown in \fig~\ref{fig:antenna}. 
The received signal is assumed to coincide with the worst-case \gls{AOA}, $\phim$ at time $t=0$. To be able to compare the sum-\gls{SNR} for different distances, we plot it without normalization. The transmitted power is set to $P_{\mathrm{t}}=23$~dBm, and the noise power at the receiver is set to $\sigma_{\mathrm{n}}^2=-95$~dBm, which correspond to the commercial IEEE802.11p radio values (at $6$~Mbit/s)~\cite{Cohda}. In \fig~\ref{fig:sum-SNR:allthree} we show the sum-\gls{SNR} for worst-case $y\in[0,2\pi)$, and worst-case $\Lw\in\{-4,4\}$, at $\dx\in\{10, 30, 100\}$.  

	From the figure, we see that the observations made earlier when analyzing the effects of the three quantities separately are valid when taking into account their combined effects.
	Namely, \gls{ACN} is impacted by time variation at medium and short distances. The choice of $\ph^\star$~\eqref{eq:design:rule} yields a robust performance under time variation of the dominant path. Despite that $\bar{\ph}=2\pi/KT$ yields optimal performance under non-time-varying conditions (i.e., under worst-case assumptions), it can exhibit severe loss in sum-\gls{SNR}, e.g., around $13$~dB at $\dx=30$~m, and around $16$~dB at $\dx=10$~m, when the dominant path is time-varying. Other choices of $\ph \in \mathcal{A}^\star$ are more robust than $\bar{\ph}$ and less robust than $\ph^\star$. 
	We recall that the antenna separation is set to $\da=10\lambda\approx 0.5$~m. Using lower $\da$ decreases the loss for all $\ph \in \mathcal{A}^\star$ and vice versa.
	In this paper, motivated by the low \gls{PL} at short distances we focused our analysis of \gls{ACN}/\gls{ABN} at medium to large distances. In \fig~\ref{fig:sum-SNR:allthree} we see a visual validation of this approach where even under severe loss in sum-\gls{SNR}, the system performance at $\dx=10$~m is higher than the performance at $\dx=100$~m and at $\dx=30$~m for the major part of the speed interval. 
	
\glsreset{PL}
\section{Conclusion and Future Works}


\Gls{ACN}~\cite{ACN} and \gls{ABN}~\cite{lehocine2021abn} are robust multiple antenna schemes that maximize the sum-\gls{SNR} of $K$ consecutive broadcast, periodic packets, under a worst-case propagation scenario, of a single dominant path with non-varying \gls{AOA}, \gls{PL}, and phase shift between antennas $\om$, over the duration of $K$ packets. In this work, we investigated the performance of a $1\times 2$ \gls{ACN} ($2\times 1$ \gls{ABN}) when the three quantities of the dominant component are time-varying instead. The main findings of this work follow.
\begin{itemize}
	\item The phase slopes that yield identical optimal sum-\gls{SNR} when the \gls{AOA}, the path-loss, and $\Omega$ are non-varying, yield different sum-\gls{SNR} when any of theses three quantities vary over the duration of $K$ packets. 
	\item Time variation of $\om$ results in shifting the \gls{ACN} phase slope, and incurs a loss in sum-\gls{SNR} that depends on the phase slope used, how fast the time variation of $\om$ is, and the antenna separation. In particular, the smaller the antenna separation, the less susceptible \gls{ACN} is to these variations.
	\item To mitigate the loss in sum-\gls{SNR} under variation of $\om$ at medium to large distances, we proposed a design rule~\eqref{eq:design:rule} that yields a phase slope that is robust against time variations and optimal under time-invariant conditions.
	\item Time variation of the \gls{PL} does not induce a shift to the \gls{ACN} phase slope, but it attenuates the sum-\gls{SNR} achieved using any phase slope compared to the sum-\gls{SNR} achieved under time-invariant conditions. The sum-\gls{SNR} loss incurred is minor compared to the loss incurred due to time variation of $\om$. The derived design rule~\eqref{eq:design:rule} is found to yield a phase slope that is robust against variation of \gls{PL} at medium to large distances, as well.
	\item Variation of \gls{AOA} induces a variation of phase and amplitude of the antenna far-field functions, which has equivalent effects to the combined $\om$ and \gls{PL} time variation effects, with an extent that depends on the antennas employed. The design rule~\eqref{eq:design:rule} yields a robust phase slope in $\Astar$ in this case too.
\end{itemize}
In this study, we investigated the performance of \gls{ACN}/\gls{ABN}  after relaxing the worst-case propagation scenario assumed when designing them. The next step is to study the performance of \gls{ACN}/\gls{ABN} under rich multipath propagation, which represents a full relaxation of the worst-case propagation assumption and is left for future work. 



	\begin{appendix}\label{appendix}
	\subsection{Preliminaries}
	 Here we present preliminary statements and lemmas that are used in the proof of Lemma~\ref{lemma:loss} and Lemma~\ref{lemma:PL}.
	 
	 Define the sets
	 \begin{align}\label{eq:appen:X}
	  	\mathcal{X}&\triangleq\{ q\pi, q\in \mathbb{Z} \},\\
	  \mathcal{X}^\star&\triangleq\{q\pi/K, q\in \mathbb{Z}\}\setminus \mathcal{X}.\label{eq:appen:Xstar}
	 \end{align}
	It follows that 
	\begin{align}
		x\in \mathcal{X} &\iff x + \pi \in \mathcal{X},\label{eq:appen:basic:per:X}\\
		(\pi/2-x) \in \mathcal{X}  &\iff (\pi/2+ x) \in \mathcal{X}.\label{eq:appen:basic:sym:X}
	\end{align}
	The first statement follows since if $x\in \mathcal{X}$ then $x= m\pi$, $m\in \mathbb{Z}$ and hence $x + \pi = (m + 1) \pi \in \mathcal{X}$. Similarly, if $(x+\pi)\in \mathcal{X}$ then $x+\pi = m\pi$, $m\in \mathbb{Z}$, and hence $x= (m-1)\pi \in \mathcal{X}$. To show~\eqref{eq:appen:basic:sym:X}, let 
	$(\pi/2\pm x)\in \mathcal{X}$, then $\pi/2 \pm x= m\pi$, where $m\in \mathbb{Z}$. Then multiplying both sides by $-1$ and adding $\pi$, we obtain $\pi/2 \mp x =(1-m)\pi $, hence $(\pi/2 \mp x)\in \mathcal{X}$.
	
	For fixed positive integer $K>1$ define the functions $f_1: x\in \mathbb{R}\setminus\mathcal{X}\rightarrow \mathbb{R}$, and $f_2: x\in \mathbb{R}\setminus\mathcal{X}\rightarrow \mathbb{R}$ following
		\begin{align}
		f_1 (x) &\triangleq   \frac{\sin(Kx)}{\sin(x)},  \label{eq:appen:f1} \\
		f_2(x)  &\triangleq   \frac{K\cos(Kx)}{\sin(x)}     -\frac{\sin(Kx)}{\sin(x)}\cot(x) . \label{eq:appen:f2}
		\end{align}
	\begin{lemma}\label{lem:appen:f1} Let $f_1$ be as defined in~\eqref{eq:appen:f1}, then
		$|f_1(x)|$ and $f_1^2(x)$ are periodic with period $\pi$, and symmetric around $\pi/2$.
	\end{lemma}
\begin{proof}
	From~\eqref{eq:appen:f1}, and $\sin(K\pi)=0$ it follows that 
	\begin{align}
		f_1 (x+\pi) = \frac{\sin(Kx+K\pi)}{\sin(x+\pi)} = -\frac{\sin(Kx)\cos(K\pi) }{\sin(x)}. \nonumber
	\end{align}
	Hence $|f_1(x+\pi)|=|f_1(x)|$ and $|f_1^2(x+\pi)|=|f_1^2(x)|$, and the periodicity property follows.
	To prove the symmetry property, we show that
	\begin{align}
		|f_1 (\pi/2+ x)|= |f_1 (\pi/2- x)|.
	\end{align}
	We have $|f_1 (\pi/2+ x)|$ is given by
	\begin{align}
\bigg|\frac{\sin(K\frac{\pi}{2})\cos(Kx)+\cos(K\frac{\pi}{2})\sin(Kx)}{\cos (x)}\bigg|. \label{eq:appen:sym:1}
	\end{align}
	On the other hand, $|f_1 (\pi/2- x)|$ is given by
		\begin{align} \bigg|\frac{\sin(K\frac{\pi}{2})\cos(Kx)-\cos(K\frac{\pi}{2})\sin(Kx)}{\cos (x)}\bigg|.\label{eq:appen:sym:2}
		\end{align}
	If $K$ is even, it follows from~\eqref{eq:appen:sym:1} and~\eqref{eq:appen:sym:2} that $|f_1 (\pi/2+ x)|= |\sin(Kx)/\cos(x)|=|f_1 (\pi/2- x)|$. If $K$ is odd, it follows from the same equations that $|f_1 (\pi/2+ x)|= |\cos(Kx)/\cos(x)|=|f_1 (\pi/2- x)|$. Hence $|f_1(x)|$ is symmetric around $\pi/2$. Since $f_1$ is a real-valued function, $f_1^2(x)=|f_1(x)|^2$, and thus $f_1^2(x)$ is symmetric around $\pi/2$ too, and the lemma follows.
\end{proof}

\begin{lemma}\label{lem:appen:f2} Let $f_2$ be as defined in~\eqref{eq:appen:f2}, then
	$f_2^2(x)$ is periodic with period $\pi$, and symmetric around $\pi/2$.
\end{lemma}
\begin{proof}
Employing trigonometric identities and noting that $\sin(K\pi)=0$, we have
\begin{align}
f_2(x+\pi) &= \frac{K \cos(Kx)\cos(K\pi)}{\sin(x+\pi)}\nonumber\\
&\quad  -\frac{\sin(Kx)\cos(K\pi)}{\sin^2(x+\pi)}\cos(x+\pi)\nonumber\\
&= -\cos(K\pi) \nonumber \\
&\quad \times\bigg( \frac{K \cos(Kx)}{\sin(x)}- \frac{\sin(Kx)}{\sin^2(x)}\cos(x)\bigg).
\end{align}
Hence, $f_2^2(x+\pi)= (-\cos(K\pi))^2 f_2^2(x)= f_2^2(x)$, and thus $f_2^2(x)$ is periodic with period $\pi$.
To prove the symmetry property, we first let $K$ be even. Using trigonometric identities and the fact that $\sin(K\frac{\pi}{2})=0$, we can express $f_2(\pi/2+x)$ as
\begin{align*}
\frac{K\cos(K\frac{\pi}{2})\cos(Kx)}{\cos(x)}+ \frac{ \cos(K\frac{\pi}{2})\sin(Kx)}{\cos^2(x)} \sin(x),
\end{align*}
and that is equal to $f_2(\pi/2-x)$, when $K$ is even.
Second, let $K$ be odd, then $f_2(\pi/2+x)$ can be expressed as
\begin{align*}
-\frac{K\sin(K\frac{\pi}{2})\sin(Kx)}{\cos(x)}+ \frac{ \sin(K\frac{\pi}{2})\cos(Kx)}{\cos^2(x)} \sin(x),
\end{align*}
while $f_2(\pi/2-x)=-f_2(\pi/2+x)$. 
Combining the cases, $K$ even and $K$ odd, we can conclude that
$f_2^2(\pi/2+x)=f_2^2(\pi/2-x)$, and hence the symmetry property holds, and the lemma follows. 
\end{proof}

\subsection{Proof of Lemma~\ref{lemma:loss} and Lemma~\ref{lemma:PL}}\label{appendix:main}
We demonstrate the results of Lemma~\ref{lemma:loss} and Lemma~\ref{lemma:PL} together.
To that end, for a fixed integer $K>1$, $x\in \mathbb{R}$, and $y\in [0,2\pi)$, define
\begin{align}
J (x,y) & \triangleq \sum_{k=0}^{K-1} ( \bJ+ \aJ k T) \cos(y-2xk), \label{eq:appen:J}
\end{align}
where $( \bJ+ \aJ kT) >0$, $k=0,1,\ldots, K-1$, implying that $\bJ >0$, and $\aJ>-\bJ/(K-1)T $. 
Moreover, define
\begin{align}
    L(x) &\triangleq -\inf_{y\in[0,2\pi)} \frac{J (x,y)}{\cJ K},\label{eq:appen:L}
\end{align}
where $\cJ = (\bJ + \aJ T(K-1)/2)>0$.
For $\bJ =\bpl$ and $\aJ =\apl$,~\eqref{eq:appen:J} and~\eqref{eq:appen:L}, correspond to the functions $\JPL$~\eqref{eq:JPL}, and $\LPL$~\eqref{eq:Loss:PL}, respectively.
Similarly, for $\bJ =1 $ and $\aJ =0$,~\eqref{eq:appen:J} and~\eqref{eq:appen:L} correspond to the functions $\JA$~\eqref{eq:J:omega} and $\LA$~\eqref{eq:loss:omega}, respectively.

To show the claims of Lemma~\ref{lemma:loss} and Lemma~\ref{lemma:PL} it is enough to show that for any $\bJ >0$ and $\aJ>-\bJ/(K-1)T $,
\begin{enumerate}
	
	\item $L(x)$ is given by
	\begin{align}\label{eq:appen:L:cases}
	L(x) =\begin{cases}
	1, & x\in\mathcal{X}\\
	\sqrt{\cJ^2f_1^2+\cJJ^2f_2^2}/(\cJ K),& x\notin\mathcal{X}
	\end{cases}
	\end{align}
	where $\cJJ = \aJ T/2$, $f_1$, $f_2$ are defined in~\eqref{eq:appen:f1},~\eqref{eq:appen:f2}, respectively, and $\mathcal{X}$ is defined in~\eqref{eq:appen:X}. 
	
	Note that when $\bJ =1 $ and $\aJ =0$, $L(x) = |f_1(x)|/K$, $x\notin\mathcal{X}$. 
	\item $L(x)$ is periodic with period $\pi$ and symmetric around $\pi/2$.
	
	\item $L(x)$ can be bounded as
	\begin{align}
		\begin{cases}\label{eq:appen:claim:3}
		0 <L(x),& \aJ \neq 0\\
		0 \leq L(x), & \aJ =0
		\end{cases}\quad,\quad L(x) \leq 1 .
	\end{align}
	\item If $x\in \mathcal{X}^\star$, where $\mathcal{X}^\star$ is defined in~\eqref{eq:appen:Xstar}, then
\begin{align}
	L(x)  = \bigg|\frac{\cJJ}{\cJ\sin(x)}\bigg|.
\end{align}	
Note that if $\aJ=0$, then $\cJJ=0$, ($\cJ \sin(x) \neq 0$) and the lower bound in~\eqref{eq:appen:claim:3} is achieved, $L(x)=0$.

	\item For $x\notin \mathcal{X}$, $L(x)$ can be bounded as
	\begin{align}
	\frac{|f_1(x)|}{K}\leq L(x) \leq \frac{\sqrt{(K-1)^2 f_1^2(x)+ f_2^2(x)}}{(K-1)K} .\label{eq:appen:L:Bounds}
	\end{align}

Note that the lower bound is achievable when $\aJ=0$.
\end{enumerate}

	\begin{proof} We start by deriving a different expression of $J$ that facilitates the proof of the lemmas. Let $x\in \mathcal{X}$, it follows that 
	\begin{align}
	    J(x,y) =\cos(y) \sum_{k=0}^{K-1} (\bJ+\aJ k T)= \cos(y)\cJ K.
	\end{align}
	Then, let $ x\notin \mathcal{X}$.	We can write
		\begin{align}\label{eq:appen:J:series}
		J(x,y) &= \bJ~ \mathrm{Re}\bigg\{ \mathrm{e}^{\jmath y} \sum_{k=0}^{K-1}\mathrm{e}^{-\jmath 2xk}  \bigg \} \nonumber \\ 
		&\quad +\aJ T~ \mathrm{Re}\bigg\{ \mathrm{e}^{\jmath y} \frac{\jmath}{2}\frac{d}{dx} \sum_{k=0}^{K-1} \mathrm{e}^{-\jmath 2xk}  \bigg \}.
		\end{align}
		Using the sum of the geometric series, we obtain 
		\begin{align}
		\sum_{k=0}^{K-1}\mathrm{e}^{-\jmath 2xk} &= \mathrm{e}^{-\jmath(K-1)x}f_1(x),\nonumber \\
		\frac{d}{dx} \sum_{k=0}^{K-1} \mathrm{e}^{-\jmath 2xk} &= -\jmath(K-1) \mathrm{e}^{-\jmath(K-1)x} f_1(x)  \nonumber\\
		&\qquad+  \mathrm{e}^{-\jmath(K-1)x} f_2(x),\nonumber 
	\end{align}
	since $f_2(x) =\frac{d}{dx} f_1(x)$.
	Substituting in~\eqref{eq:appen:J:series} we arrive at
		\begin{align}
		J(x,y) &= \begin{cases}
		\cJ K ~\cos (y), & x \in  \mathcal{X}\\
		\cJ f_1 (x) ~\cos\big(y- (K-1) x\big)  \\\quad -\cJJ f_2(x) ~\sin \big(y- (K-1) x\big), & x\notin \mathcal{X}
		\end{cases}\label{eq:appen:J:cases}
		\end{align}
		Now we tackle the lemmas claims.
		\begin{enumerate}
			\item To show that $L(x)$ is given by~\eqref{eq:appen:L:cases}, we first let $x\in \mathcal{X}$. Then, it follows from~\eqref{eq:appen:J:cases} that
			\begin{align*}
				L(x)=-\inf_{y\in[0,2\pi)} \frac{J(x, y)}{\cJ K} = -\inf_{y\in[0,2\pi)}  \cos (y) = 1.
			\end{align*}	
			Second, let $x\notin \mathcal{X}$, and write 
			\begin{align}
			&\cJ f_1 (x) = R(x) \cos\big( \varphi(x) \big),\nonumber \\
			-&\cJJ  f_2(x) = R(x) \sin\big( \varphi(x) \big),\nonumber\\
			&R(x) =\sqrt{\cJ^2f_1^2(x)+ \cJJ^2f_2^2(x)},\nonumber\\
			& \varphi(x) =\arctan \bigg(-\frac{\cJJ f_2(x)}{\cJ f_1 (x)} \bigg).\nonumber
			\end{align}
			Then, it follows from~\eqref{eq:appen:J:cases} that 
			\begin{align*}
			J(x,y) = R(x) \cos\big(y- (K-1) x - \varphi(x)\big),
			\end{align*}
			and thus, when $x\notin \mathcal{X}$
			\begin{align}
			-\inf_{y\in [0,2\pi)} \frac{J(x,y)}{\cJ K} = \frac{R(x)}{\cJ K} = \frac{\sqrt{\cJ^2f_1^2(x)+ \cJJ^2f_2^2(x)}}{\cJ K}.\nonumber
			\end{align}
			where the minimizer is given by $y = \rem \big((K-1) x + \varphi(x)+\pi, 2\pi\big)$.
			Hence,~\eqref{eq:appen:L:cases} holds.
			\item To show that $L$ is periodic and symmetric, let $x \in \mathcal{X}$, then from~\eqref{eq:appen:L:cases} we get $L(x)=1$, which is a constant.  Using~\eqref{eq:appen:basic:per:X} and~\eqref{eq:appen:basic:sym:X} we can conclude that $L$ is periodic with period $\pi$, and symmetric around $\pi/2$ when $x\in \mathcal{X}$.  The same statement is true when $x\notin \mathcal{X}$, since by~\eqref{eq:appen:L:cases} $L(x)=\sqrt{\cJ^2f_1^2(x)+ \cJJ^2f_2^2(x)}/(\cJ K)$, and since both $f_1^2$ and $f_2^2$ are periodic and symmetric as shown in Lemma~\ref{lem:appen:f1} and Lemma~\ref{lem:appen:f2}, respectively.
		
			\item To show~\eqref{eq:appen:claim:3}, we first proof that $L(x)\leq 1$. To that end, observe that $\sum_{k=0}^{K-1} (\bJ+\aJ kT) = \cJ K$. Then using~\eqref{eq:appen:J} we deduce that
			\begin{align}
			-\cJ K\leq &J(x,y) \leq \cJ K.
			\end{align}
			Then, it follows that $L(x)=-\inf_{y}J(x,y)/(\cJ K)\leq 1 $, which is achievable when $x\in \mathcal{X}$.
			
			Second, to show the lower bounds in~\eqref{eq:appen:claim:3}, we employ~\eqref{eq:appen:L:cases} and the fact that $\cJ>0$ to deduce that
			\begin{align}
			L(x) =0 &\iff x\notin \mathcal{X} \nonumber \\
			        &\iff \big(f_1^2(x)=0 , \cJJ^2f_2^2(x) =0\big) .
			\end{align}
			
			Since $f_1(x)=0$ if and only if $x\in \mathcal{X}^\star$, and since 
			\begin{align}
			f_2^2(x) = \bigg ( K\frac{\cos(Kx)}{\sin(x)}\bigg)^2 >0,\quad x\in \mathcal{X}^\star,
			\end{align}
			we conclude that 
			$0<L(x)$ when $\aJ\neq 0$, (since $\cJJ= \aJ T/2$), and $0\leq L(x)$, when $\aJ= 0$.
			
			\item Let $x\in \mathcal{X}^\star$, from the definition~\eqref{eq:appen:Xstar} we deduce that $x=m\pi/K$, where $m\neq qK$, $q\in\mathbb{Z}$, hence $x\notin\mathcal{X}$. Substituting in~\eqref{eq:appen:L:cases}, and recalling from the previous argument that $f_1(x)=0$ we get
			\begin{align}
			L(x)=\frac{|\cJJ f_2(x)|}{\cJ K}&=\bigg|\cJJ K\frac{\cos(m\pi)}{\cJ K\sin(x)}\bigg|\nonumber \\
			&=\bigg|\frac{\cJJ}{\cJ\sin(x)}\bigg|.
			\end{align}

			\item From \eqref{eq:appen:L:cases} it follows that, for $x\notin\mathcal{X}$, 
        	\begin{align}
	        L(x) =
	            \frac{1}{K}\sqrt{f_1^2(x) + \left(\frac{c_2}{c_1}\right)^2 f_2^2(x)}.
	        \end{align}
            From the definitions of $c_1$ and $c_2$, we deduce that 
            \begin{align}
                \frac{c_2}{c_1} 
                &=  \frac{aT/2}{b + a(K-1)T/2}\\
                &=  \frac{1}{K-1}\frac{a/b}{a/b + 2/(K-1)T}\\
                &=  \frac{1}{K-1}\frac{w}{w + 2A}, \label{eq:c2c1}
            \end{align}
            where $w=a/b$ and $A=1/(K-1)T$. From the condition $a> -b/(K-1)T$, we have that \eqref{eq:c2c1} is valid for $w > -1/(K-1)T = -A$. 
            
            We can now write, for $x\notin\mathcal{X}$, 
        	\begin{align}
	        L(x; w) 
	        &= \frac{1}{K}\sqrt{f_1^2(x) + g(w)\frac{f_2^2(x)}{(K-1)^2}}\\ 
	        &= \frac{\sqrt{(K-1)^2f_1^2(x) + g(w) f_2^2(x)}}{K(K-1)}, 
	        \end{align}
            where 
            \begin{align}
                g(w) = \left(\frac{w}{w + 2A}\right)^2,\qquad w > -A.
            \end{align}
            It is easily seen that, for a fixed $x\neq\mathcal{X}$, $L(x;w)$  increases as $g(w)$ increases. Moreover, it is easily verified that
            \begin{align*}
                g(0) &= 0,\\
                \lim_{w\to -A} g(w) &= \lim_{w\to\infty} g(w) = 1,\\
                \frac{d}{d w} g(w) &=
                \frac{4A w}{(w+2A)^3}=
                \begin{cases}
                    \le 0, & -A< w \le 0\\
                    > 0, & w > 0
                \end{cases}
            \end{align*}
			In other words, $g(w)$ is decreasing for $-A< w \le 0$, attains its minimum as $w=0$, and increases for $w\ge 0$. It follows that $\sup_{w>-A} g(w) = 1$ and $\inf_{w>-A} g(w) = 0$, and, therefore, 
            \begin{align}
	        L(x; w) 
	        &\le \sup_{w > -A} L(x, w)\\
	        &= \frac{\sqrt{(K-1)^2 f_1^2(x) + f_2^2(x)}}{K(K-1)}, 
	        \end{align}
        	and 
        	\begin{align}
	        L(x; w) 
	        &\ge \inf_{w > -A} L(x, w) = \frac{|f_1(x)|}{K}. 
	        \end{align}
    		Hence~\eqref{eq:appen:L:Bounds} holds. 
		\end{enumerate}
		
		All the claims have been shown and thus Lemma~\ref{lemma:loss}, and Lemma~\ref{lemma:PL} follow.
	\end{proof}

\end{appendix}

\bibliographystyle{ieeetr}
\bibliography{refIII.bib}

\end{document}